\documentclass[aps, showpacs,preprintnumbers, superscriptaddress, nofootinbibt, twocolumn]{revtex4-2}
\usepackage{amssymb}
\usepackage{graphicx}
\usepackage{makeidx}
\usepackage{amsfonts}
\usepackage{amsthm}
\usepackage{amsmath}

\usepackage[hyperindex,breaklinks]{hyperref}
\usepackage{eurosym}
\usepackage{hyperref}\hypersetup{
    colorlinks = true,
    linkcolor = blue,
    anchorcolor = blue,
    citecolor = blue,
    filecolor = blue,
    urlcolor = blue,
    breaklinks=true
   }
\usepackage{color}
\usepackage{revsymb}
\usepackage[dvipsnames]{xcolor}
\usepackage{tensor}

\newtheorem{definition}{Definition}[section]
\newtheorem{remark}{Remark}[section]
\newtheorem{theorem}{Theorem}[section]
\newtheorem{corollary}{Corollary}[section]
\newtheorem{proposition}{Proposition}[section]

\def\be{\begin{equation}}
\def\ee{\end{equation}}
\def\bea{\begin{eqnarray}}
\def\eea{\end{eqnarray}}

\def\be{\begin{equation}}
\def\ee{\end{equation}}
\def\bea{\begin{eqnarray}}
\def\eea{\end{eqnarray}}

\begin{document}

\title{Semi-Symmetric Metric Gravity: from the Friedmann-Schouten geometry with torsion to dynamical dark energy models}
\author{Lehel Csillag}
\email[e-mail:]{lehel@csillag.ro}
\affiliation{Department of Physics, Babes-Bolyai University, Kogalniceanu Street, Cluj Napoca, 400084, Romania}
\affiliation{Ludwig Maximillian University, Theresienstr. 37, 8033 Munich, Germany}
\author{Tiberiu Harko}
\email[e-mail:]{tiberiu.harko@aira.astro.ro}
\affiliation{Department of Physics, Babes-Bolyai University, Kogalniceanu Street, Cluj Napoca, 400084, Romania}
\affiliation{Department of Theoretical Physics, National Institute of Physics
and Nuclear Engineering (IFIN-HH), Bucharest, 077125 Romania,}
\affiliation{Astronomical Observatory, 19 Ciresilor Street, Cluj-Napoca 400487, Romania}

\begin{abstract}

    In the present paper we introduce a geometric generalization of standard general relativity, based on a geometry initially introduced by Friedmann and Schouten in 1924, through the notion of a semi-symmetric connection. The semi-symmetric connection is a particular connection that extends the Levi-Civita one, by allowing for the presence of torsion. While the mathematical landscape of the semi-symmetric metric connections is well-explored, their physical implications remain to be investigated.  After presenting in detail the differential geometric aspects of the geometries with semi-symmetric metric connection, we formulate the Einstein field equations, which contain additional terms induced by the presence of the specific form of torsion we are studying. We consider the cosmological applications of the theory by deriving the generalized Friedmann equations in a flat, homogeneous and isotropic geometry. The Friedmann equations also include some supplementary terms as compared to their general relativistic counterparts, which can be interpreted as a geometric type dark energy. To evaluate the proposed theory, we consider three cosmological models - the first with constant effective density and pressure,  the second with the dark energy satisfying a linear equation of state, and a third one one with the effective quantities satisfying a polytropic equation of state. We also compare the predictions of the semi-symmetric metric gravitational theory with the observational data for the Hubble function, and with the predictions of the standard $\Lambda$CDM model. Our findings indicate that the semi-symmetric metric cosmological models give a good description of the observational data, and for certain values of the model parameters, they can reproduce almost exactly the predictions of the $\Lambda$CDM paradigm. Consequently, Friedmann's initially proposed geometry emerges as a credible alternative to standard general relativity, in which dark energy has a purely geometric origin.
\end{abstract}

\date{\today}
\maketitle

{
  \hypersetup{linkcolor=blue}
  \tableofcontents
}

\newpage
\section{Introduction}
\hypersetup{citecolor=blue}

General relativity is a cornerstone of modern theoretical physics. The theory proposed by Einstein and Hilbert \cite%
{Einstein,Hilbert,Einstein2} describes spacetime as a four-dimensional smooth manifold, equipped with a semi-Riemannian metric $g_{\mu \nu}$. Gravity is interpreted as the described by the Ricci tensor $R_{\mu \nu}$ and Ricci scalar $R$, obtained as contractions of the curvature tensor $\tensor{R}{^\mu_\nu_\rho_\lambda}$ associated to the Levi-Civita connection $\nabla$. The field equations relate the Einstein tensor $G_{\nu \lambda}$, which can be obtained from $R_{\mu \nu}$ and $R$, and the energy-momentum tensor $T_{\nu \lambda}$, which describes ordinary baryonic matter.

The theory based on  semi-Riemannian geometry provides a precise description of  gravitational physics at the level of the Solar System, including effects such as  light deflection, perihelion precession of Mercury and the Shapiro time delay \cite{Will}. Recently, the existence of gravitational waves, which have been predicted by general relativity in the twentieth century, have also been experimentally confirmed \cite{Gravwaves}.

Einstein's geometric theory of gravity not only shaped physicists' understanding of gravity and successfully explained observational data, but also had an immense impact on pure mathematics. Thanks to its geometric foundations, it attracted the attention of mathematicians, including Weyl and Cartan, who attempted to generalize the theory by also extending Riemannian geometry. They quickly realized that the Levi-Civita connection plays a crucial role in Einstein's theory. This connection is unique, and is characterized by metric-compatibility $\nabla_\mu g_{\nu \rho}=0$ and torsion-freeness: $\tensor{T}{^\mu_\nu _\rho}=0$, which corresponds to the fact that the Christoffel symbols $\tensor{\Gamma}{^{\mu}_{\nu \rho}}$ are symmetric. Therefore, with a generalization in mind, these two ingredients could be modified: Weyl considered weakening the metric-compatibility condition \cite{Weyl, Weyl1}, while Cartan  introduced torsion \cite{Cartan1}, which was a new concept back then for differential geometers.

Initially, Weyl's goal was to find a unified field theory, which would describe gravity and electromagnetism. With this in mind, he considered a non-metricity $Q_{\mu \nu \rho}=-\nabla_{\mu} g_{\nu \rho} \neq 0$ of a special form, where the trace part of the non-metricity is the Weyl vector. Physically, the Weyl vector $\omega_{\mu}$ was interpreted to be the source of the electromagnetic field. Although Weyl's theory was able to reproduce the classical predictions of the simpler theory of general relativity (perihelion advance of Mercury, gravitational redshift), Einstein pointed out that the identification would imply that identical atoms which move in such a way as to enclose some electromagnetic flux would change size after the motion. In turn, different sized atoms would have different spectral lines, which has not been experimentally measured. This has lead to the abandonment of the theory by the physics community. Despite the rejection of his theory, Weyl was convinced that his geometry could be useful, perhaps through a reinterpretation of the Weyl vector. This can be seen from his reply to Einstein:

"Your rejection of the theory for me is weighty; [...] But my own brain still keeps believing in it. And as a mathematician I must by all means hold to [the fact] that my geometry is the true geometry 'in the near', that Riemann happened to come to the special case $F_{ik}=0$ is due only to historical reasons (its origin is the theory of surfaces), not to such that matter." \cite{Collectedpapers}

The mathematician \'{E}lie Cartan tried to extend general relativity from a different angle, by introducing the concept of torsion. In this way, he created a theory, which is nowadays called  Einstein-Cartan theory  \cite%
{Cartan2,Cartan3,Cartan4}. Here, the torsion $\tensor{T}{^\mu_\nu _\rho}$ is assumed to be proportional to the spin-density of matter. The Einstein-Cartan theory has been extensively investigated in the physical literature \cite{Hehl, Hehl1,Hehl2,Hehl3}. The Weyl geometry can be easily generalized to include  torsion. The corresponding  geometry is called the Weyl-Cartan geometry, and its mathematical properties and physical applications have been studied in depth \cite{WC1,WC2,WC3,WC4,WC5,WC6,WC7,WC8}.

A third important mathematical, as well as physical, development is related to the work of Weitzenb\"{o}ck \cite{Weitz},  who constructed the theory of  the so-called Weitzenb\"{o}ck spaces. A Weitzenb\"{o}ck space is a manifold with properties $\nabla _{\mu }g_{\sigma \lambda }= 0$, $T^{\mu }_{\sigma \lambda }\neq 0$, and $R^{\mu }_{\nu \sigma \lambda }=0$, where $g_{\sigma \lambda }$, $T^{\mu }_{\sigma \lambda }$ and $R^{\mu }_{\nu \sigma \lambda }$ are the metric, the torsion, and the curvature tensors of the manifold, respectively. If $T^{\mu }_{\sigma \lambda }= 0$, the Weitzenb\"{o}ck manifold becomes a Euclidean manifold. Due to the fact that their Riemann curvature tensor is zero, Weitzenb\"{o}ck spaces have the property of distant parallelism, also called as absolute, or teleparallelism. Einstein first proposed to use Weitzenb\"{o}ck type geometries in a unified teleparallel theory of electromagnetism and gravity \cite{Ein}. In the teleparallel approach one considers, as basic physical variables, the set of tetrad vectors $e^i_{\mu }$. The torsion, obtained from the tetrad fields, can then be used to obtain an alternative description of general relativity,
since the  curvature can be substituted by  torsion. In this way we obtain the teleparallel equivalent of General Relativity (TEGR) \cite{TP1}, also called the $f(T)$ gravity model, in which  torsion exactly compensates curvature, with the space-time becoming flat. TEGR theories are also an active field of research, and they have the potential of explaining the observational problems of modern cosmology \cite{TP2, TP3,TP4,TP5,TP6,TP7,TP8,TP9,TP10,TP11}. For a recent review of $f(T)$ gravity see \cite{TPRev}.

Unified field theories were not appealing only to mathematicians like Weyl, but they also attracted the attention of Schr\"{o}dinger \cite{Schrod} and Eddington \cite{Edd} as well, who considered non-Riemannian connections as a route to unification. For a detailed review of the history of unified field theories, see \cite{unified}. Despite the fact that the unification was not successful, Schr\"{o}dinger \cite{Schrod} found the most general affine connection, which preserves the lengths of vectors. This connection is remarkable in the sense that it overcomes Einstein's objections of Weyl's theory, as the length of vectors is preserved under parallel transport.

Despite being largely set aside in the twentieth century, the majority of the aforementioned extensions have been recently reevaluated, prompted by the theoretical challenges recently encountered by general relativity. Most of these challenges come from  recent observational data, which cannot be explained using standard general relativity without adding extra parameters to the theory.

The first major challenge is the observation of the accelerated expansion of the Universe, which is confirmed by  data of the Planck satellite \cite{Planck}, and studies of Baryon Acoustic Oscillations \cite{Baryon1,Baryon2,Baryon3}. For a detailed review of the cosmic acceleration problem, see \cite{acceleration}.

A closely related, but different problem of modern cosmology is that Einstein's theory does not properly describe the galactic dynamics of massive test particles around the galactic center, which can be seen from the astrophysical observations of the galactic rotation curves \cite{rotationcurves1, rotationcurves2, rotationcurves3}. To explain the observed data, the existence of dark matter is assumed. In most of the theoretical models, it is described as a pressureless cosmic fluid, but its origin is not yet known. There are several particles, which were proposed to be dark matter candidates, such as sterile neutrinos or axions. See \cite{darkmattercandidate1,darkmattercandidate2} for reviews of  dark matter particle candidates.

To address the above mentioned challenges, in modern day cosmology, a simple and widely accepted theory, the $\Lambda$CDM model was formulated. It includes, outside of ordinary matter, a cosmological constant $\Lambda$, which is associated with dark energy, and a postulated cold dark matter. Despite the success of this model, it introduces extra parameters, whose physical origin remains unknown. For instance, despite decades of experimental efforts, there is yet no conclusive evidence supporting the existence of dark matter particles. This has led to the exploration of alternative gravity theories. An intriguing possibility is that the observed phenomena, typically attributed to the presence of dark matter, might actually be a manifestation of the gravitational force. In this perspective, general relativity might break down at large scales, and might have to be replaced by a generalized theory, which potentially explains the accelerated expansion of the Universe, and the astrophysical effects associated with dark matter.

In this context, as mentioned earlier, numerous gravitational theories discarded in the twentieth century have gained renewed interest. Rather than being viewed as unified field theories, they are now considered as alternative theories of gravity, with the potential to explain observational data. This trend, similar to the situation in the twentieth century, has lead to an interplay between mathematics and physics through the study of the cosmological and astrophysical implications of  non-Riemannian geometries.

A prominent example of this reconsideration is  Weyl geometry. Nonmetricity is the basic geometrical/physical quantity in the $f(Q)$ gravity theory,
first proposed in \cite{fQ1}, and later generalized in \cite{fQ2}.  In the $f(Q)$ gravitational theory, the non-metricity $Q$, having its origins in a Weyl type geometric background, is used to describe the properties of the gravitational interaction in a flat geometry, with vanishing curvature. $f(Q)$ gravity was widely used to investigate the cosmological evolution of the Universe, and stellar properties \cite{fQ3,fQ4,fQ5,fQ6,fQ7,fQ8,fQ9,fQ10,fQ11}.

A novel perspective on the physical applications of the Weyl geometry was introduced in the papers \cite{Gh1,Gh2,Gh3,Gh4,Gh5}, where a spin zero, scalar degree of freedom was extracted by linearizing the action with the help of an auxiliary scalar field. This led to the formulation of the so-called Weyl geometric gravity theory.   Its astrophysical implications have been extensively studied in \cite{Harkoweylastro}, where an exact black hole solution that generalizes the Schwarzschild solution has also been found.  Weyl geometric gravity also gives an alternative explanation for the observed data of test particles outside of galaxies, which is usually explained, as we have noted, by assuming the existence of dark matter. In this sense, if dark matter is not found, its effects could be understood by the motion of particles in a Weyl type geometry. For a detailed discussion of this approach, see \cite{Harkodarkmatter, Harkodarkmattera}. The applications of  Weyl geometric gravity  for the description of the dynamical evolution of the Universe were investigated in \cite{Gh6,Gh7}.

In a recent paper \cite{HarkoSchr}, the physical and cosmological implications of a gravitational theory based on a Schr\"{o}dinger connection are thoroughly investigated. In this article, the authors consider a Schr\"{o}dinger connection, a connection that conserves the length of the vectors, and which is implemented by non-metricity only. The corresponding gravitational theory was derived from a gravitational Lagrangian of the form
$L=R+\frac{5}{24}Q_\rho Q^\rho+\frac{1}{6}\tilde{Q}_\rho\tilde{Q}^\rho+2T_\rho Q^\rho
+\zeta^{\rho\sigma}_{~~\alpha}T^\alpha_{~\rho\sigma}$, where $Q_\rho$, $\tilde{Q}_\mu$, and $T_\nu$ are the nonmetricity and the torsion vectors. $\zeta^{\rho\sigma}_{~~\alpha}$ are Lagrange multipliers.

Klemm and Ravera \cite{SilkeKlemm} pointed out that the Schr\"{o}dinger connections could be realized not only by  non-metricity, but also by torsion, or both non-metricity and torsion. This remark has lead us to the consideration of finding Schr\"{o}dinger connections which are constructed with the help of torsion.

It is not widely known, but one year before his death, and shortly after he wrote down the famous Friedmann equations, in $1924$ A. Friedmann, together with Schouten,  introduced a new connection \cite{FriedmannSchouten}, the semi-symmetric connection. This connection  is characterized by its torsion being determined by a one-form. It took almost half a century, until this connection has been studied by mathematicians. In 1970, Kentaro Yano \cite{Kentaroyano} investigated the geometric properties of the semi-symmetric metric connections, which are simply semi-symmetric connections with the condition $\nabla_{\mu} g_{\nu \rho}=0$. On the other hand, the physics literature on semi-symmetric connections is brief. Surprisingly, the physical applications of semi-symmetric connections are highlighted in a mathematical paper \cite{zangiabadi}, which quotes Schouten's book on Ricci calculus \cite{Riccicalculusbook}. A possible application is the description of a specific type of displacement; if a person is moving on the surface of the Earth always facing one definite point, then this displacement is semi-symmetric, and metric. A different physical application of semi-symmetric connections is given in \cite{unificationsemisym}, where an attempt to unify gravity and electromagnetism is made, based on semi-symmetric connections. However, such a unification is possible only in the two dimensional case.

The role of the torsion in cosmology was emphasized in \cite{Kranas}, where a torsion tensor of the form $S_{abc}=2\phi h_{a[b}u_{c]}$, has been adopted, with $h_{ab}$ denoting the projection tensor $h_{ab}=g_{ab}+u_au_b$. The continuity, Friedmann and the Raychaudhuri equations were derived for this specific form of torsion. Moreover, it was shown that torsion alone can induce exponential expansion, also leading to a reduction of the helium-4 production in the early Universe. In the framework of Einstein-Cartan theory, cosmological implications of torsion have also been investigated in \cite{Ben1,Ben2,Luz}.

History demonstrates that attempts to formulate geometric unified field theories have not been successful. Hence, the goal of the present paper is to reconsider semi-symmetric metric connections from a different perspective, with a focus on its viability as a geometric extension of standard general relativity. As a first step to achieve this goal, we write down the Einstein equations
\begin{equation}
    R_{(\mu \nu)}-\frac{1}{2} g_{\mu \nu} R=8 \pi T_{\mu \nu}
\end{equation}
in this geometry, where $R_{\mu \nu}$ and $R$ represent the Ricci tensor and Ricci scalar of the semi-symmetric metric connection. Due to the presence of torsion, we will obtain five new terms in the Einstein equations, which we interpret as a geometric type dark energy, a point which will be justified later.

Having the generalized Einstein equations at hand, for the cosmological applications, we examine the corresponding generalized Friedmann equations within this geometry. The presence of torsion in the connection results in the inclusion of two extra terms in the first Friedmann equation, and three additional terms in the second Friedmann equation. We interpret the supplementary terms in the first Friedmann equation as an effective energy density attributed to dark energy, while the additional terms in the second Friedmann equation represent an effective pressure associated with dark energy. To assess the viability of the theory, we examine three distinct cosmological models. In the first model, we assume that the effective dark energy and dark pressure terms are constants. This model can be fully solved analytically, and leads to a simple theoretical alternative to the $\Lambda$CDM model.   In the second model, we assume that dark energy satisfies a linear equation of state, implying that the effective geometric pressure is proportional to the energy density of the dark energy. The parameter of the equation of state is considered to be a redshift-dependent function. In the third model, we assume that the dark energy satisfies a polytropic equation of state, with polytropic index $n=1$. For models II and III, we conduct a comparative analysis of the predictions derived from semi-symmetric metric cosmology with a small dataset containing observational data for the Hubble function, and with the standard $\Lambda$CDM paradigm, respectively. Our findings indicate that the semi-symmetric metric cosmological models offer a  robust description of observational data for the Hubble function. Additionally, for certain values of the model parameters, these models can closely replicate the predictions of the $\Lambda$CDM model. Hence, semi-symmetric gravity turns out to be a viable alternative to standard general relativity, presenting a scenario where dark energy has a purely geometric origin.

The present paper is organized as follows. We introduce the basics of  semi-symmetric connections in Section~\ref{section2}, and relate them to Murgescu's generalized Weyl connections \cite{murgescu} and to  Schr\"{o}dinger connections. After having seen that semi-symmetric connections can be recast as Schr\"{o}dinger connections, the Einstein equations in this geometry are presented. In Section~\ref{section3}, the cosmological applications of the semi-symmetric metric gravity are investigated in detail. The generalized Friedmann equations are obtained, together with an exact analytical de Sitter solution. To compare the theoretical predictions of the model with observational data, the dimensionless and redshift representation are worked out. Three distinct cosmological models in semi-symmetric metric gravity are taken into consideration in Section \ref{section4}, whose predictions are compared both with observational data for the Hubble function and the standard $\Lambda$CDM model.  A summary of results and an outlook is given in section \ref{section6}. Notably, in the present paper we adopt a different convention, than the one often used for the torsion tensor. Hence, for the sake of completeness, we present the decomposition of the affine connection in our convention in Appendix \ref{appendixA}.
 Appendix \ref{appendixC} is devoted to the algebraic calculations used to derive the Friedmann equations in  semi-symmetric metric geometry. Finally, to highlight the interplay between mathematics and physics, in  Appendix \ref{appendixD} we present a coordinate-free approach to semi-symmetric connections, which is attributed to Kentaro Yano \cite{Kentaroyano}. In addition to the coordinate-free presentation, we explain in detail how one can obtain the formulae in coordinates, which appear in the article, proving that the equations used in our investigations indeed do follow from the coordinate-free abstract mathematical approach.

 \section{Semi-Symmetric Metric Gravity} \label{section2}

In this Section, we present the relation between semi-symmetric metric connections, Schr\"{o}dinger connections, and generalized Weyl connections. Then, as an application of the presented results, we compute the Einstein equations in the semi-symmetric metric geometry. We start our investigations with the introduction of the basic mathematical theory of the semi-symmetric connections.

\subsection{Basic concepts of semi-symmetric geometry}\label{basicconcepts}

In their groundbreaking paper \cite{FriedmannSchouten}, Friedmann and Schouten characterized a semi-symmetric connection as an affine connection, where the torsion tensor satisfies
\begin{equation}\label{semisymmetric}
    \tensor{T}{^\mu_\nu _\rho}=\pi_{\rho} \delta^{\mu}_{\nu} - \pi_{\nu} \delta^{\mu}_{\rho}\; \; \text{for some one-form} \; \; \pi _\rho.
\end{equation}

This directly implies that the torsion tensor, when expressed with lower indices, takes the form
\begin{equation}\label{torsiondown}
    T_{\lambda \nu \rho}=g_{\lambda \mu} \tensor{T}{^\mu _\nu _\rho}=\pi_{\rho} g_{\lambda \nu} - \pi_{\nu} g_{\lambda \rho}.
\end{equation}

Furthermore, when the non-metricity is null, the connection is called \textbf{semi-symmetric metric}. We work in the convention, where the torsion and non-metricity of an affine connection are defined as
\begin{equation}
    \tensor{T}{^\mu _\nu _\rho}:=2 \tensor{\Gamma}{^{\mu}_{[\rho \nu]}},\; \; Q_{\mu \nu \rho}:=-\nabla_{\mu} g_{\nu \rho}.
\end{equation}

We motivate this convention for the torsion tensor as follows. It can be obtained immediately from the coordinate-free definition
\begin{equation}
\begin{aligned}
    T:&\Gamma(T^{*}M) \times \Gamma(TM) \times \Gamma(TM) \to C^{\infty}(M),\\
    &\; \; T(\omega,X,Y)=\omega( \nabla_X Y - \nabla_Y X - [X,Y]).
\end{aligned}
\end{equation}
The components can be obtained from this geometric definition by
\begin{equation}
    \tensor{T}{^\mu _\nu _\rho}=T(dx^\mu,\partial_\nu, \partial_\rho)=dx^{\mu} \left(\nabla_{\partial _\nu} \partial_\rho - \nabla_{\partial_\rho} \partial_\nu -[\partial_\nu,\partial_\rho] \right).
\end{equation}

Using the definition of Christoffel symbols and that partial derivatives commute leads to
\begin{equation}
    \tensor{T}{^\mu _\nu _\rho}=dx^\mu \left(\tensor{\Gamma}{^\beta_{\rho \nu}} \partial_{\beta} - \tensor{\Gamma}{^\beta _\nu _\rho} \partial_{\beta} \right).
\end{equation}
Since the covector field $dx^{\mu}$ is $C^{\infty}(M)$-multilinear
\begin{equation}
     \tensor{T}{^\mu _\nu _\rho}= \tensor{\Gamma}{^\beta_{\rho \nu}} dx^{\mu}(\partial_\beta) -\tensor{\Gamma}{^\beta _\nu _\rho} dx^{\mu} (\partial_\beta).
\end{equation}
As $dx^{\mu}$ is the dual basis of $\partial_{\beta}$, i.e. $dx^\mu(\partial_\beta)=\delta^{\mu}_{\beta}$, we obtain the desired result
\begin{equation}
    \tensor{T}{^\mu _\nu _\rho}=\tensor{\Gamma}{^\mu _\rho _\nu} - \tensor{\Gamma}{^\mu _\nu _\rho}= 2 \tensor{\Gamma}{^\mu _{[\rho \nu]}}.
\end{equation}

It is shown in Appendix~\ref{appendixA}, that in this convention an affine connection can be decomposed using the non-metricity and the torsion as
\begin{equation}\label{generalconnection}
\begin{aligned}
    \tensor{{\Gamma}}{^\mu _\nu _\rho}=\tensor{\gamma}{^\mu _\nu _\rho} &+ \frac{1}{2} g^{\lambda \mu}(-Q_{\lambda \nu \rho}+ Q_{\rho \lambda \nu} + Q_{\nu \rho \lambda})\\
    &- \frac{1}{2}g^{\lambda \mu}(T_{\rho \nu \lambda}+T_{\nu \rho \lambda}- T_{\lambda \rho \nu}),
    \end{aligned}
\end{equation}
where $\tensor{\gamma}{^\lambda _\mu _\nu}$ denotes the Christoffel symbols of the Levi-Civita connection. To obtain the Christoffel symbols of a semi-symmetric metric connection, we set the non-metricity to be zero, and consider the torsion to be given by \eqref{torsiondown}. In this case, Eq.~\eqref{generalconnection} reduces to
\begin{equation}
\begin{aligned}
      \tensor{{\Gamma}}{^\mu _\nu _\rho}= \tensor{\gamma}{^\mu _\nu _\rho} -&\frac{1}{2}g^{\lambda \mu} \left( \pi_{\lambda} g_{\rho \nu} - \pi_{\nu} g_{\rho \lambda} \right)\\
      -&\frac{1}{2} g^{\lambda \mu}( \pi_{\lambda} g_{\nu \rho} - \pi_{\rho} g_{\nu \lambda})\\
      -&\frac{1}{2} g^{\lambda \mu} (\pi_{\rho} g_{\lambda \nu} -\pi_{\nu} g_{\lambda \rho}).
      \end{aligned}
\end{equation}
An immediate simplification of the above equations yields
\begin{equation}
       \tensor{{\Gamma}}{^\mu _\nu _\rho}= \tensor{\gamma}{^\mu _\nu _\rho} - g^{\mu \lambda} \pi_{\lambda} g_{\rho \nu} + g^{\lambda \mu} \pi_{\nu} g_{\rho \lambda},
\end{equation}
or equivalently
\begin{equation}\label{Christoffelsemisymmetric}
     \tensor{{\Gamma}}{^\mu _\nu _\rho}= \tensor{\gamma}{^\mu _\nu _\rho} - \pi^{\mu} g_{\rho \nu} + \pi_{\nu} \delta^{\mu}_{\rho}.
\end{equation}
Our convention for the Riemann curvature tensor is the following
\begin{equation}
\begin{aligned}
    \tensor{Riem}{^\mu _\nu _\rho _\sigma}&=\tensor{\Gamma}{^\lambda _\nu _\sigma} \tensor{\Gamma}{^\mu _\lambda _\rho} - \tensor{\Gamma}{^\lambda _\nu _\rho} \tensor{\Gamma}{^\mu_\lambda _\sigma}\\
   &+ \partial_{\rho} \tensor{\Gamma}{^\mu _\nu _\sigma} - \partial_{\sigma} \tensor{\Gamma}{^\mu _\nu _\rho}
   \end{aligned},
\end{equation}
while the Ricci tensor and scalar are obtained as
\begin{eqnarray}
    R_{\nu \sigma}=\tensor{Riem}{^\mu_\nu_\mu_\sigma}, \; \; R=g^{\nu \sigma} R_{\nu \sigma}, \; \; \text{respectively}.
\end{eqnarray}
In Appendix \ref{appendixD}, using Kentaro Yano's coordinate-free results, we derive a local formula for the Riemann tensor of a semi-symmetric metric connection. Hence, the Riemann tensor, according to Corollary \ref{corollaryriemann}, is expressed in local coordinates as:
\begin{equation}\label{riemanncurvaturesemisym}
  \begin{aligned}
    \tensor{Riem}{^\mu_{\nu \rho \sigma}}=&\overset{\circ}{Riem} \tensor{}{^\mu _\nu _\rho _\sigma}- S_{\sigma \nu} \delta^{\mu}_{\rho}+S_{\rho \nu} \delta^{\mu}_{\sigma}\\
    &- g_{\sigma \nu} S_{\rho \lambda} g^{\lambda \mu}+g_{\rho \nu} S_{\sigma \lambda} g^{\lambda \mu},
\end{aligned}
\end{equation}
where $\overset{\circ}{Riem}\tensor{}{^\mu _\nu _\rho _\sigma}$ denotes the Riemann tensor of the Levi-Civita connection and
\begin{equation}\label{stensorpitensor}
S_{\nu \sigma}=\overset{\circ}{\nabla}_{\nu} \pi_{\sigma}- \pi_{\nu} \pi_{\sigma} + \frac{1}{2} g_{\nu \sigma} \pi_{\lambda} \pi^{\lambda}.\end{equation}
A straightforward computation shows that
\begin{equation}\label{Riccicurvaturesemisym}
    R_{\nu \sigma}=\overset{\circ}{R}_{\nu \sigma} -2 S_{\sigma \nu} -g_{\nu \sigma} S_{\lambda \beta} g^{\lambda \beta}.
\end{equation}
By contracting with $g^{\nu \sigma}$, one obtains the Ricci scalar of the semi-symmetric metric connection
\begin{equation}\label{Ricciscalarsemisym}
    R=\overset{\circ}{R} - 6 S_{\beta \lambda} g^{\beta \lambda}.
\end{equation}
\subsection{From semi-symmetric connections to generalized Weyl connections}

Up until now, we have discussed semi-symmetric connections on Lorentzian manifolds. The idea can be carried over to Weyl manifolds as well, however in that case we cannot implement metric-compatibility, due to the presence of non-metricity in Weyl's theory. In the case of Weyl geometry, we have a non-metricity of the form
\begin{equation}
    Q_{\mu \nu \rho}= -2w_{\mu} g_{\nu \rho}.
\end{equation}

By considering this expression of the  non-metricity, and the torsion tensor to be of the form \eqref{semisymmetric}, we obtain a semi-symmetric Weyl connection. Its Christoffel symbols can be obtained from Eq.~\eqref{generalconnection} by substituting the choices for non-metricity and torsion
\begin{equation}
\begin{aligned}
    \tensor{\Gamma}{^\mu_\nu_\rho}=\tensor{\gamma}{^\mu_\nu_\rho} &+ \frac{1}{2} g^{\lambda \mu} \left(2 w_\lambda g_{\nu \rho} -2 w_\rho g_{\lambda \nu} - 2 w_\nu g_{\rho \lambda} \right)\\
    &- \pi^{\mu} g_{\rho \nu} + \pi_{\nu} \delta^{\mu}_{\rho}.
\end{aligned}
\end{equation}

A straightforward simplification yields
\begin{equation}\label{Weylsemisym}
    \tensor{\Gamma}{^\mu_\nu_\rho}=\underbrace{\tensor{\gamma}{^\mu_\nu_\rho} + w^{\mu} g_{\nu \rho} - w_{\rho} \delta^{\mu}_{\nu} - w_{\nu} \delta^{\mu}_{\rho}}_{:=\overset{W}{\Gamma}\tensor{}{^\mu _\nu _\rho}} - \pi^{\mu} g_{\rho \nu} + \pi_{\nu} \delta^{\mu}_{\rho}.
\end{equation}
By introducing the notation $\overset{W}{\Gamma}\tensor{}{^\mu _\nu _\rho}$ for the Christoffel symbols of the Weyl connection, the Christoffel symbols of the semi-symmetric Weyl connection take the following compact form:
\begin{equation}\label{weylsemisym2}
      \tensor{\Gamma}{^\mu_\nu_\rho}=\overset{W}{\Gamma}\tensor{}{^\mu _\nu _\rho} - \pi^{\mu} g_{\rho \nu} + \pi_{\nu} \delta^{\mu}_{\rho}.
\end{equation}

In this way, we have successfully introduced torsion into Weyl geometry. From a mathematical point of view, work in this direction has already been done by Murgescu \cite{murgescu}, in a more general setting, who defined the notion of a generalized Weyl connection. This connection is characterized by the following Christoffel symbols
\begin{equation} \label{murgescuconnection}
    \tensor{\Gamma}{^\mu _\nu _\rho}=\overset{W}{\Gamma}\tensor{}{^\mu _\nu _\rho} +a_{\nu \rho \lambda} g^{\lambda \mu},
\end{equation}
where the tensor $a_{\nu \rho \lambda}$ is given by
\begin{equation}
    a_{\nu \rho \lambda}=g_{\nu \beta} \tensor{T}{^\beta _\rho _\lambda}+ g_{\beta \rho } \tensor{T}{^\beta _\nu _\lambda} + g_{\beta \lambda} \tensor{T}{^\beta _ \nu _\rho}.
\end{equation}
\"{U}nal and Uysal \cite{unaluysal} noted that by choosing
\begin{equation}
    \tensor{T}{^\mu _\nu _\rho}=\delta^{\mu}_{\nu} a_{\rho} - \delta^{\mu}_{\rho} a_{\nu}
\end{equation}
one obtains a semi-symmetric Weyl connection satisfying \eqref{weylsemisym2} with $\pi_{\rho}=-2a_{\rho}$. In the subsequent section we present the detailed argument, which leads to this result.
\subsection{Schr\"{o}dinger, generalized  Weyl, and semi-symmetric connections}\label{appendixB}

We now consider  the mathematical details on how to relate generalized Weyl, Schr\"{o}dinger, and semi-symmetric connections.

\subsubsection{Weyl and Murgescu connections}

First, we show how one can obtain a semi-symmetric Weyl connection \eqref{weylsemisym2}, starting from Murgescu's connection \eqref{murgescuconnection}. We first choose the torsion tensor, as proposed in \cite{unaluysal}
    \begin{equation}
    \tensor{T}{^\mu _\nu _\rho}=\delta^{\mu}_{\nu} a_{\rho} - \delta^{\mu}_{\rho} a_{\nu}.
\end{equation}

Substituting this form of torsion into \eqref{murgescuconnection} leads to
\begin{equation}
\begin{aligned}
    \tensor{\Gamma}{^\mu _\nu _\rho}=\overset{W}{\Gamma}\tensor{}{^\mu _\nu _\rho}&+ g^{\lambda \mu}  g_{\nu \beta} \left(\delta^{\beta}_{\rho} a_{\lambda}- \delta^{\beta}_{\lambda} a_{\rho} \right) \\
    &+ g^{\lambda \mu} g_{\beta \rho}\left(\delta^{\beta}_{\nu} a_{\lambda} - \delta^{\beta}_{\lambda} a_{\nu} \right)\\
    &+ g^{\lambda \mu} g_{\beta \lambda} \left(\delta^{\beta}_{\nu} a_{\rho} - \delta^{\beta}_{\rho}a_{\nu} \right)
\end{aligned}
\end{equation}

Simplifying the Kronecker deltas gives
\begin{equation}
    \begin{aligned}
         \tensor{\Gamma}{^\mu _\nu _\rho}=\overset{W}{\Gamma}\tensor{}{^\mu _\nu _\rho}&+ g^{\lambda \mu} \left( g_{\nu \rho} a_{\lambda} - g_{\nu \lambda} a_{\rho} \right)\\
         &+g^{\lambda \mu} \left(g_{\nu \rho} a_{\lambda}- g_{\lambda \rho} a_{\nu} \right)\\
         &+ g^{\lambda \mu} \left(g_{\nu \lambda} a_{\rho} - g_{\rho \lambda} a_{\nu} \right)
    \end{aligned}
\end{equation}

One immediately observes that the terms containing $g_{\nu \lambda} a_{\rho}$ vanish, and after regrouping  the other terms it follows that
\begin{equation}
 \tensor{\Gamma}{^\mu _\nu _\rho}=\overset{W}{\Gamma}\tensor{}{^\mu _\nu _\rho}+ 2g^{\lambda \mu} \left( g_{\nu \rho} a_{\lambda} - g_{\lambda \rho} a_{\nu} \right).
\end{equation}

A straightforward simplification yields
\begin{equation}
     \tensor{\Gamma}{^\mu _\nu _\rho}=\overset{W}{\Gamma}\tensor{}{^\mu _\nu _\rho}+ 2 a^{\mu} g_{\nu \rho} - 2 \delta^{\mu}_{\rho} a_{\nu}.
\end{equation}

Upon choosing $\pi^{\mu}:=-2 a^{\mu}$, we arrive at the desired result
\begin{equation}
    \tensor{\Gamma}{^\mu _\nu _\rho}=\overset{W}{\Gamma}\tensor{}{^\mu _\nu _\rho} - \pi^{\mu} g_{\nu \rho} + \delta^{\mu}_{\rho} \pi_{\nu},
\end{equation}
which is in perfect accordance with \eqref{weylsemisym2}. This shows that starting from a generalized Weyl connection \eqref{murgescuconnection}, one can obtain a semi-symmetric Weyl connection \eqref{weylsemisym2}.

\subsubsection{Semi-symmetric and Schr\"{o}dinger connections}

We now turn our attention to the relation between Schr\"{o}dinger and semi-symmetric connections.
A Schr\"{o}dinger connection has the form
\begin{equation}
    \tensor{\Gamma}{^\mu_\nu _\rho}= \tensor{\gamma}{^\mu_\nu _\rho}+g^{\mu \lambda} Z_{\lambda \nu \rho},
\end{equation}
where $Z_{\lambda \nu \rho}$ is a tensor, which satisfies
\begin{equation}\label{conditions}
    Z_{\lambda \nu \rho}=Z_{\lambda \rho \nu}, \; \; Z_{(\lambda \nu \rho)}=0.
\end{equation}

This connection is remarkable in the sense that it overcomes Einstein's objection to Weyl's geometry: the second condition $Z_{(\lambda \nu \rho)}=0$ ensures that vectors preserve their length under parallel transport.

We will now prove that using a torsion of semi-symmetric form, we can define a Schr\"{o}dinger connection as
\begin{equation}\label{schrconn}
    Z_{\lambda \nu \rho}=\frac{1}{2} \left(\pi_{\lambda} g_{\nu \rho} - \pi_{\rho} g_{\nu \lambda} + \pi_{\lambda} g_{\nu \rho} - \pi_{\nu} g_{\rho \lambda} \right).
\end{equation}

To see that $Z_{\lambda \nu \rho}$, defined by  \eqref{schrconn}  is a Schr\"{o}dinger connection, we verify the two conditions. The first condition reads
\begin{equation}
    Z_{\lambda \nu \rho}=Z_{\lambda \rho \nu}.
\end{equation}

As a first step of verifying the equality, we write out the right hand side explicitly:
\begin{equation} \label{schrconn2}
    Z_{\lambda \rho \nu}=\frac{1}{2} \left(\pi_{\lambda} g_{\rho \nu} - \pi_{\nu} g_{\rho \lambda} + \pi_{\lambda} g_{\rho \nu} - \pi_{\rho} g_{\nu \lambda} \right).
\end{equation}

The symmetry of the metric tensor, namely
\begin{equation}
    g_{\nu \rho}=g_{\rho \nu}
\end{equation}
directly implies that the right hand sides of \eqref{schrconn2} and \eqref{schrconn} are equal. Consequently, the first condition is satisfied. We now turn to proving the second condition
\begin{equation}
    Z_{(\lambda \nu \rho)}=0.
\end{equation}

This can be equivalently rewritten as
\begin{equation}
    \frac{1}{6} \left(\underbrace{Z_{\lambda \nu \rho}}_{\text{term }1}+ \underbrace{Z_{\rho \lambda \nu}}_{\text{term } 2} + \underbrace{Z_{\nu \rho \lambda}}_{\text{term } 3} \right)=0.
\end{equation}

We write out term $1$, term $2$ and term $3$ separately
\begin{equation}\label{term 1}
   Z_{\lambda \nu \rho}=\frac{1}{2} \left( \textcolor{red}{\pi_{\lambda} g_{\nu \rho}} - \textcolor{blue}{\pi_{\rho} g_{\nu \lambda}} + \textcolor{orange}{\pi_{\lambda} g_{\nu \rho}} -\textcolor{violet}{\pi_{\nu} g_{\rho \lambda}} \right),
\end{equation}
\begin{equation}\label{term 2}
    Z_{\rho \lambda \nu}=\frac{1}{2} \left(\textcolor{blue}{ \pi_{\rho} g_{\lambda \nu}} - \textcolor{teal}{\pi_{\nu} g_{\lambda \rho}} + \textcolor{RubineRed}{\pi_{\rho} g_{\lambda \nu}}- \textcolor{red}{\pi_{\lambda} g_{\nu \rho}} \right),
\end{equation}
\begin{equation}\label{term 3}
    Z_{\nu \rho \lambda}=\frac{1}{2} \left( \textcolor{violet}{ \pi_{\nu} g_{\rho \lambda}} - \textcolor{orange}{\pi_{\lambda} g_{\rho \nu}} + \textcolor{teal}{\pi_{\nu} g_{\rho \lambda}} - \textcolor{RubineRed}{\pi_{\rho} g_{\lambda \nu}}\right).
\end{equation}

Adding term $1$, term $2$ and term $3$, that is, \eqref{term 1} $+$ \eqref{term 2} $+$ \eqref{term 3}, everything cancels. This is highlighted by the colors. Hence, the desired equality
\begin{equation}
    Z_{(\lambda \nu \rho)}=0
\end{equation}
is obtained, proving that Schr\"{o}dinger connections can be implemented via torsion of semi-symmetric type.

\subsection{A torsional Schr\"{o}dinger connection}

In the following, with the help of the semi-symmetric connections, we will provide an example for a  Schr\"{o}dinger connection, which is implemented by torsion, in contrast to the Schr\"{o}dinger connection implemented by non-metricity in \cite{HarkoSchr}.

To explain the relation between semi-symmetric metric and Schr\"{o}dinger connections, let us suppose that we have a semi-symmetric metric connection, with torsion tensor
\begin{equation}
    \tensor{T}{^\mu _\nu _\rho}=\pi_{\rho} \delta^{\mu}_{\nu} - \pi_{\nu} \delta^{\mu}_{\rho}.
\end{equation}

By Eq.~\eqref{torsiondown}, it follows that
\begin{equation}
    T_{\lambda \nu \rho}=\pi_{\rho} g_{\lambda \nu}-\pi_{\nu} g_{\lambda \rho}.
\end{equation}

We now claim that if we choose
\begin{equation}
    Z_{\lambda \nu \rho}=T_{(\nu \rho) \lambda},
\end{equation}
with our choice of torsion, i.e.
\begin{equation}\label{Schr}
    Z_{\lambda \nu \rho}=\frac{1}{2} \left(\pi_{\lambda} g_{\nu \rho} - \pi_{\rho} g_{\nu \lambda} + \pi_{\lambda} g_{\nu \rho} - \pi_{\nu} g_{\rho \lambda} \right),
\end{equation}
gives a Schr\"{o}dinger connection, which is implemented purely by torsion. This means, that the tensor $Z_{\lambda \nu \rho}$, defined in \eqref{Schr}, satisfies the conditions given in Eq.~\eqref{conditions}. This can be verified  by some straightforward algebra, which can be found in Section IIC. We note that if we would not have assumed that the semi-symmetric connection is also metric, the Schr\"{o}dinger connection could have been realized by both torsion and non-metricity.

The exposition on semi-symmetric connections is now complete. We have demonstrated their close association with both generalized Weyl connections, and Schr\"{o}dinger connections with torsion. In the subsequent Section, our focus shifts to calculating the generalized Einstein equation in this geometry.

\subsection{The generalized Einstein field equations}

To derive the generalized Einstein equations in semi-symmetric metric geometry, we use the formulas for the Ricci curvature and Ricci scalar obtained in Section \ref{basicconcepts}. First of all, {\it we postulate} that the generalized Einstein equation is given by
\begin{equation}
    R_{(\nu \sigma)} - \frac{1}{2} R g_{\nu \sigma}=8 \pi T_{\nu \sigma},
\end{equation}
where $R_{\nu \sigma}$ and $R$ are the Ricci curvature and scalar of the semi-symmetric metric connection, respectively, and $T_{\nu \sigma}$ denotes the matter energy-momentum tensor. Substituting \eqref{Riccicurvaturesemisym} and \eqref{Ricciscalarsemisym} yields
\begin{equation}
\begin{aligned}
    \overset{\circ}{R}_{\nu \sigma} &- 2S_{(\sigma \nu)} - g_{\nu \sigma} S_{\lambda \beta} g^{\lambda \beta}\\
    &- \frac{1}{2} g_{\nu \sigma} \left(\overset{\circ}{R} - 6 S_{\beta \lambda} g^{\beta \lambda} \right)=8 \pi T_{\nu \sigma}.
\end{aligned}
\end{equation}

Upon simplifying and rearranging terms, this can be expressed as
\begin{equation}\label{einstein}
    \overset{\circ}{R}_{\nu \sigma}- \frac{1}{2} g_{\nu \sigma} \overset{\circ}{R}-S_{\sigma \nu} - S_{\nu \sigma} +2 g_{\sigma \nu} S_{\lambda \beta} g^{\lambda \beta}=8\pi T_{\nu \sigma}.
\end{equation}

From the above equation it is readily seen that in the limit when  $S$ becomes zero, we recover the usual Einstein equation.

Eq.~\eqref{stensorpitensor} shows that there is an explicit formula relating $S_{\nu \sigma}$ to the one-form $\pi$ of the semi-symmetric connection.

Using this formula, we will derive the Einstein equations purely in terms of the one-form $\pi$. We start by substituting equation ~\eqref{stensorpitensor} into equation ~\eqref{einstein}
\begin{equation}
\begin{aligned}
    &\overset{\circ}{R}_{\nu \sigma}- \frac{1}{2} g_{\nu \sigma} \overset{\circ}{R}- \left(\overset{\circ}{\nabla}_{\sigma} \pi_{\nu} -\pi_{\sigma} \pi_{\nu} + \frac{1}{2}  g_{\sigma \nu} \pi^{\rho} \pi_{\rho} \right)\\
    &-\left(\overset{\circ}{\nabla}_{\nu} \pi_{\sigma} - \pi_\nu \pi_\sigma +\frac{1}{2} g_{\sigma \nu} \pi^\rho \pi_\rho \right)\\
    &+2 g_{\sigma \nu} \left(\overset{\circ}{\nabla}_{\lambda} \pi_\beta  - \pi_\lambda \pi_\beta + \frac{1}{2} g_{\lambda \beta} \pi^\rho \pi_\rho\right) g^{\lambda \beta}=8 \pi T_{\nu \sigma}.
\end{aligned}
\end{equation}
Collecting the terms and contracting the last bracket with $g^{\lambda \beta}$ yields
\begin{equation}
    \begin{aligned}
       &\overset{\circ}{R}_{\nu \sigma}- \frac{1}{2} g_{\nu \sigma} \overset{\circ}{R}-\overset{\circ}{\nabla}_{\sigma} \pi_{\nu} - \overset{\circ}{\nabla}_{\nu} \pi_{\sigma} +2 \pi_{\sigma} \pi_{\nu}- g_{\sigma \nu} \pi^{\rho} \pi_\rho\\
       &+2 g_{\sigma \nu} \left(\overset{\circ}{\nabla}_{\lambda} \pi^{\lambda} - \pi_{\lambda} \pi^{\lambda} +2 \pi^{\rho} \pi_{\rho} \right)=8 \pi T_{\nu \sigma}.
    \end{aligned}
\end{equation}
Hence, our field equations take the final form
\begin{equation}\label{Einsteinsemisymmetricequation}
\begin{aligned}
    \overset{\circ}{R}_{\nu \sigma}&-\frac{1}{2} g_{\nu \sigma} \overset{\circ}{R} - \overset{\circ}{\nabla}_{\sigma} \pi_{\nu} - \overset{\circ}{\nabla}_{\nu} \pi_{\sigma}\\
    &+ 2 \pi_{\sigma} \pi_{\nu}+ 2 g_{\sigma \nu} \overset{\circ}{\nabla}_{\lambda} \pi^{\lambda} + g_{\nu \sigma} \pi^{\rho}\pi_{\rho}=8\pi T_{\nu \sigma}.
\end{aligned}
\end{equation}
This derivation shows, that in addition to the usual Einstein equation, there are five terms, all of which contain the one-form $\pi$ of the semi-symmetric connection. We interpret these terms as a geometric-type dark energy. We will justify this interpretation by considering cosmological applications of the theory. As a first step towards the justification, we derive the generalized Friedmann equations from Eq.~\eqref{Einsteinsemisymmetricequation}, and consider three cosmological models. By comparing the predictions of the linear and polytropic models  with the standard $\Lambda$CDM paradigm, it will be shown that the additional terms arising from the one-form $\pi$ could be an alternative explanation for the observational data, which are usually attributed to the presence of the dark energy.

\section{Cosmological applications}\label{section3}

This Section is devoted to the detailed study of the viability of the semi-symmetric metric gravity as an alternative theory of gravity, by comparing its cosmological predictions to the standard $\Lambda$CDM paradigm, and to the observational data. We begin this exploration by deriving the Friedmann equations in this theory.

\subsection{The Friedmann equations}\label{Friedmannsection}

The starting point to derive the Friedmann equations is the  generalized Einstein equation \eqref{Einsteinsemisymmetricequation}. We consider an isotropic, homogeneous and spatially flat FLRW metric, which is described by the metric
\begin{equation}
    ds^2=-dt^2+a^2(t) \delta_{ij} dx^i dx^j,
\end{equation}
where the indices $i,j$ are spatial indices, i.e. they take the values $1,2,3$. Furthermore, for matter we consider a perfect fluid, characterized by two thermodynamic parameters, the pressure $p$ and the energy density $\rho$. Its energy-momentum tensor reads
\begin{equation}
    T_{\nu \sigma}=\rho u_{\nu} u_{\sigma} + p(u_\nu u_\sigma + g_{\nu \sigma}).
\end{equation}
We look at the problem in a comoving coordinate system, in which the four-velocity takes the form
\begin{equation}
    u_\nu=(-1,0,0,0) \iff u^{\nu}=(1,0,0,0).
\end{equation}

Moreover, we consider the ansatz
\begin{equation}
    \pi_{\nu}=(-\omega(t),0,0,0) \iff \pi^{\nu}=(\omega(t),0,0,0).
\end{equation}

From now on, we do not write out explicitly the time dependence for $a(t)$. Given these assumptions, some straightforward algebra detailed in Appendix \ref{appendixC} leads to the following Friedmann equations
\begin{eqnarray}  \label{F1}
\hspace{-0.5cm}3H^{2}=8\pi \rho -3 \omega^2 +6H \omega=8\pi \left( \rho
+\rho _{eff}\right) =8\pi \rho _{tot},
\end{eqnarray}
\begin{eqnarray}  \label{F2}
2\dot{H}+3H^{2}&=&-8\pi p+4H\omega -\omega ^{2}+2\dot{\omega}=-8\pi \left(
p+p_{eff}\right)  \notag \\
&=&-8\pi p_{tot},
\end{eqnarray}
where we have denoted $H=\dot{a}/a$, and
\begin{equation}
\rho _{eff}=\frac{1}{8 \pi} \left( 6H \omega -3 \omega^2 \right) ,
\end{equation}%
and
\begin{equation}
p_{eff}=-\frac{1}{8\pi }\left( 4H\omega -\omega ^{2}+2\dot{\omega}\right) ,
\end{equation}%
respectively, while $\rho _{tot}=\rho +\rho _{eff}$, and $p_{tot}=p+p_{eff}$.

By multiplying the first Friedmann equation with $a^{3}$, by taking the
derivative of the resulting equation with respect to the time, and with
the use of the second Friedmann equation, we obtain the energy conservation
equation
\begin{equation}
\dot{\rho}_{tot}+3H\left( \rho _{tot}+p_{tot}\right) =0,
\end{equation}%
which can be rewritten in an equivalent form as
\begin{equation}
\dot{\rho}+3H\left( \rho +p\right) +\dot{\rho}_{eff}+3H\left( \rho
_{eff}+p_{eff}\right) =0,
\end{equation}%
or
\bea\label{eqcons}
\hspace{-0.5cm}\dot{\rho}&+&3H\left( \rho +p\right) \nonumber\\
\hspace{-0.5cm}&+&\frac{3}{8\pi }\left[ \frac{d}{dt}\left(
2H\omega -\omega ^{2}\right) +H\left( 2H\omega -2\omega ^{2}-2\dot{\omega}\right) %
\right] =0. \nonumber\\
\eea

Eq.~(\ref{F1}) can be reformulated as $3\left(H-\omega\right)^2=8\pi \rho$, giving
\be\label{66}
H=\omega \pm \sqrt{\frac{8\pi}{3}}\sqrt{\rho}.
\ee
and
\be
\dot{H}=\dot{\omega}\pm \sqrt{\frac{2\pi}{3}}\frac{\dot{\rho}}{\sqrt{\rho}}.
\ee
Hence, the second Friedmann equation gives for the evolution of the density the equation
\be\label{67}
\pm \dot{\rho}\pm 2H\rho+\sqrt{\frac{8\pi}{3}}\left(\rho+3p\right)\sqrt{\rho}=0.
\ee

As an indicator of the accelerating/decelerating nature of the cosmological
expansion we use the deceleration parameter, defined as
\begin{equation}
q=\frac{d}{dt}\frac{1}{H}-1=-\frac{\dot{H}}{H^{2}}-1.
\end{equation}

With the use of the generalized Friedmann equations we obtain
\begin{equation}
q=\frac{1}{2}+\frac{3}{2}\frac{p_{tot}}{\rho _{tot}}=\frac{1}{2}+\frac{3}{2}%
\frac{8\pi p-\left( 4H\omega -\omega ^{2}+2\dot{\omega}\right) }{8\pi \rho
+3\left( 2H\omega -\omega ^{2}\right) }.
\end{equation}

\subsection{de Sitter type evolutions}

We investigate now the de Sitter type solutions of the metric semi-symmetric gravity theory. We consider the possibilities of both vacuum and non-vacuum evolutions.

\subsubsection{Vacuum solutions}

In the vacuum,  with $\rho =p=0$, the first Friedmann equation (\ref{66}) immediately gives
\begin{equation}
    H(t)=\omega(t),
\end{equation}
while the second Friedmann equation (\ref{67}) is automatically satisfied. Hence, the de Sitter type solution with $H=H_0={\rm constant}$ is obtained
for $\omega(t)=\omega_0=\text{constant}$, that is, a vacuum Universe in the presence of a constant semi-symmetric constant torsion tensor will expand exponentially according to
\begin{equation}
    \frac{\dot a}{a}=\omega_0 \implies a(t)=a_0 e^{\omega_0 t}.
\end{equation}

\subsubsection{de Sitter expansion in the presence of matter}

In the presence of dust matter with $p=0$, a de Sitter type expansion with $H=H_0={\rm constant}$ is obtained if the matter density satisfies the differential equations
 \be\label{67}
\pm \dot{\rho}\pm 2H_0\rho+\sqrt{\frac{8\pi}{3}}\rho ^{3/2}=0,
\ee
with the general solution given by
\be
\rho (t)=\frac{4 H_0^4}{\left(\mp e^{H_0(t-C_0)}+2 \sqrt{\frac{2\pi}{3}}H_0 \right)^2}
\ee
where $C_0$ is an arbitrary integration constant. Hence, during the de Sitter type evolution, the ordinary matter density decreases exponentially, and tends to zero in the large time limit, $\lim_{t\rightarrow \infty }\rho (t)=0$.  As for the torsion component $\omega$, its behavior is obtained as
\be
\omega (t)=H_0\mp \sqrt{\frac{8\pi}{3}}\frac{2 H_0^2}{\left(\mp e^{H_0(t-C_0)}+2 \sqrt{\frac{2\pi}{3}}H_0 \right)}.
\ee

In the limit of large times $\omega (t)$ tends to $H_0$, thus becoming a constant, and hence there is a smooth transition between the matter and the vacuum de Sitter type expansions.

\subsection{The dimensionless representation}

In order to simplify the mathematical formalism, we introduce a set of
dimensionless variables $\left( h,\tau ,\Omega ,r,P\right) $, defined
according to
\begin{equation}
H=H_{0}h,\tau =H_{0}t,\omega =H_{0}\Omega ,\rho =\frac{3H_{0}^{2}}{8\pi }r,p=%
\frac{3H_{0}^{2}}{8\pi }P.
\end{equation}

Then the generalized Friedmann equations and the energy balance equation of the semi-symmetric metric gravity theory take the
dimensionless form
\begin{equation}
h^{2}=r- \Omega^2+ 2 h \Omega,
\end{equation}
\begin{equation}
2\frac{dh}{d\tau }+3h^{2}=-3P+4h\Omega -\Omega ^{2}+2\frac{d\Omega }{d\tau }.
\end{equation}%
The effective energy density and pressure become
\begin{equation}
r_{eff}=2h\Omega -\Omega ^{2},
\end{equation}

\begin{equation}
P_{eff}=-\frac{1}{3}\left( 4h\Omega -\Omega ^{2}+2\frac{d\Omega }{d\tau }%
\right) ,
\end{equation}
where $\rho _{eff}=\left( 3H_{0}^{2}/8\pi \right) r_{eff}$, and $%
p_{eff}=\left( 3H_{0}^{2}/8\pi \right) P_{eff}$.

\subsection{The redshift representation}

To directly compare the theoretical predictions of the model with the
observational data we introduce, instead of the time variable, the redshift $%
z$, defined according to $1+z=1/a$. Then we obtain
\begin{equation}
\frac{d}{d\tau }=-(1+z)h(z)\frac{d}{dz}.
\end{equation}

Thus, in the redshift representation, the Friedmann equations are given by
\begin{equation}
h^{2}(z)=r(z)+2h(z)\Omega (z)-\Omega ^{2}(z),  \label{Fz1}
\end{equation}%
\begin{eqnarray}  \label{Fz2}
&&-2(1+z)h(z)\frac{dh(z)}{dz}+3h^{2}(z)=-3P(z)+4h(z)\Omega (z)  \notag \\
&&-\Omega ^{2}(z)-2(1+z)h(z)\frac{d\Omega }{dz}.
\end{eqnarray}%

To test the viability of the present cosmological model, we perform a
detailed comparison with the standard $\Lambda $CDM model, as well as with a
small sample of observational data points, obtained for the Hubble function.

In the $\Lambda $CDM model the Hubble function is obtained as
\begin{equation}
H=H_{0}\sqrt{\frac{\Omega _{m}}{a^{3}}+\Omega _{\Lambda }}=H_{0}\sqrt{\Omega
_{m}(1+z)^{3}+\Omega _{\Lambda }},
\end{equation}%
where $\Omega _{m}=\Omega _{b}+\Omega _{DM}$, with $\Omega _{b}=\rho
_{b}/\rho _{cr}$, $\Omega _{DM}=\rho _{DM}/\rho _{cr} $ and $\Omega
_{\Lambda }=\Lambda /\rho _{cr}$, where $\rho_{cr}$ is the critical density
of the Universe. $\Omega _{b}$, $\Omega _{DM}$ and $\Omega _{DE}$ represent
the density parameters of the baryonic matter, dark matter, and dark energy,
respectively. The deceleration parameter can be obtained from the relation
\begin{equation}
q(z)=\frac{3(1+z)^{3}\Omega _{m}}{2\left[ \Omega _{\Lambda }+(1+z)^{3}\Omega
_{m}\right] }-1.
\end{equation}

In the following analysis for the matter and dark energy density parameters
of the $\Lambda $CDM model we will use the numerical values $\Omega
_{DM}=0.2589$, $\Omega _{b}=0.0486$, and $\Omega _{\Lambda }=0.6911$,
respectively \cite{1g}. Hence, the total matter density parameter $%
\Omega _{m}=\Omega _{DM}+\Omega _{b}=0.3075$, where we have neglected the
contribution of the radiation to the total matter energy balance in the late
Universe. The present day value of $q$, as predicted by the $\Lambda$CDM
model, is thus $q(0)=-0.5912$, indicating that the recent Universe is in an
accelerating expansionary stage. For the observational data we use the
values of the Hubble functions from the compilation presented in \cite{Bou}.

To compare the present model with the $\Lambda$CDM paradigm we will also use
the $Om(z)$ diagnostic \cite{Sahni}, representing an important method for
the differentiation of the alternative cosmological models. The $Om(z)$
function is defined as
\begin{equation}
Om (z)=\frac{H^2(z)/H_0^2-1}{(1+z)^3-1}=\frac{h^2(z)-1}{(1+z)^3-1}.
\end{equation}

In the case of the $\Lambda$CDM model, $Om(z)$ is a constant, and it is
equal to the present day matter density $r(0)=0.3075$. For cosmological
models satisfying an equation of state with a constant equation of state
parameter $w = \mathrm{constant}$, the existence of a positive slope of $%
Om(z)$ is evidence for a phantom-like evolution, while a negative slope
indicates a quintessence-like dynamics.

\section{Specific cosmological models}\label{section4}
In the present formalism there is no equation of motion for $\Omega $.
Hence, in order to close the cosmological field equations, we need to impose
a specific relation between the parameters of the model. We will consider three such models, in which some particular relations (equations of state) are imposed on the effective energy density and pressure, generated by the torsion field. In the first model we assume that both these quantities are constants. In the second cosmological model we assume the existence of a linear, redshift dependent relation between $\rho _{eff}$ and $p_{eff}$. Finally, we consider a model in which the effective pressure is related to the effective density by a polytropic relation, with polytropic index $n=1$.

\subsection{Model I: \;$\rho_{eff}=\lambda$, $p_{eff}=2K \lambda/3$}

As a first cosmological model in metric semi-symmetric gravity theory we consider the case in which both the effective density and pressure are constants, with the constants taking different values. Hence, we impose the conditions
\begin{equation}
    \frac{1}{8\pi} \left( 6H \omega - 3 \omega^2 \right)=\lambda,
\end{equation}
\begin{equation}
 - \frac{1}{8\pi} \left( 4 H \omega - \omega^2 +2 \dot \omega \right)=\frac{2}{3} K \lambda.
\end{equation}
To get rid of signs and factors, we redefine $\Lambda=8\pi \lambda, k=-K$ to obtain
\be
3\omega \left(2H-\omega\right)=\Lambda,
\ee
and
\be
4H\omega -\omega ^2+2\dot{\omega}=\frac{2}{3}k\Lambda,
\ee
respectively, where $k$ and $\Lambda \geq 0$ are constants. By eliminating $H$ between the above two equations, we obtain for $\omega$ the first order differential equation
\be
2\dot{\omega}+\omega ^2+2(1-k)\frac{\Lambda}{3}=0,
\ee
with the general solution given by
\be
\omega (t)=\sqrt{\frac{2(k-1)\Lambda}{3}}\tanh\left[\frac{\sqrt{(k-1)\Lambda}}{\sqrt{6}}\left(t-t_0\right)\right],
\ee
where $t_0$ is an arbitrary constant of integration. For this expression of $\omega$, we obtain the Hubble function as
\begin{eqnarray}
\hspace{-0.5cm}H(t) &=&\frac{\sqrt{\Lambda }}{2\sqrt{6(k-1)}}\tanh \left[ \frac{\sqrt{%
(k-1)\Lambda }\left( t-t_{0}\right) }{\sqrt{6}}\right]  \nonumber\\
\hspace{-0.5cm}&&\times \left\{ \coth ^{2}\left[ \frac{\sqrt{(k-1)\Lambda }\left(
t-t_{0}\right) }{\sqrt{6}}\right] +2(k-1)\right\} .\nonumber\\
\end{eqnarray}

As for the matter density, $8\pi \rho =3H^2-\Lambda$, it is given by
\bea
\hspace{-0.5cm}8\pi \rho (t)&=&\frac{\Lambda }{8(k-1)}\Bigg\{ \coth \left[ \sqrt{\frac{%
(k-1)\Lambda }{6}}\left( t-t_{0}\right) \right]\nonumber\\
\hspace{-0.5cm} &&-2(k-1)\tanh \left[ \sqrt{%
\frac{(k-1)\Lambda }{6}}\left( t-t_{0}\right) \right] \Bigg\}^{2}.
\eea

 For the matter pressure $8\pi p(t)=2k\Lambda/3 -2\dot{H}-3H^2$ we obtain the expression
\bea
8\pi p(t)&=&\frac{\Lambda }{24(k-1)}\Bigg\{ (4k-7)\text{csch}^{2}\left[ \sqrt{%
\frac{(k-1)\Lambda }{6}}\left( t-t_{0}\right) \right] \nonumber\\
&&+4(k-1)^{2}\text{sech}%
^{2}\left[ \sqrt{\frac{(k-1)\Lambda }{6}}\left( t-t_{0}\right) \right]\nonumber\\
&&+4k^{2}-4k-3\Bigg\}.
\eea

The scale factor of this cosmological model is given by
\bea
a(t)&=&a_{0}\sinh ^{2(k-1)}\left[ \sqrt{\frac{(k-1)\Lambda }{6}}\left(
t-t_{0}\right) \right] \nonumber\\
&&\times \cosh ^{\sqrt{6}}\left[ \sqrt{\frac{(k-1)\Lambda }{6}}%
\left( t-t_{0}\right) \right] .
\eea

In the limit of large times we obtain for the matter density the expression
\be
\lim_{t\rightarrow \infty} \rho(t)=\frac{\left(3-2k\right)^2}{8(k-1)}\Lambda,
\ee
while in the same limit the pressure becomes
\be
\lim_{t\rightarrow \infty} p(t)=\frac{4k^2-4k-3}{24(k-1)}\Lambda.
\ee
For $k=3/2$ we have $\lim_{t\rightarrow \infty} p(t)=0$, and $\lim_{t\rightarrow \infty} \rho(t)=0$, indicating that the Universe ends in a vacuum state. For other values of $k$ the cosmological evolution ends in constant density and pressure thermodynamic phase. The deceleration parameter is given by
\bea
q(t)&=&-\frac{2(k-1)\left\{ 2(k-1)-\coth ^{2}\left[ \sqrt{\frac{(k-1)\Lambda }{%
6}}\left( t-t_{0}\right) \right] \right\} }{\left\{ \coth ^{2}\left[ \sqrt{%
\frac{(k-1)\Lambda }{6}}\left( t-t_{0}\right) \right] +2(k-1)\right\} {}^{2}}\nonumber\\
&&\times \text{csch}^{2}\left[ \sqrt{\frac{(k-1)\Lambda }{6}}\left( t-t_{0}\right) %
\right] -1,
\eea

In the large time limit $\lim _{t\rightarrow \infty}q(t)=-1$, and hence in this cosmological model the Universe ends in a de Sitter type accelerating phase.

\subsection{Model II:\;$P_{eff}=-\sigma (z)r_{eff}-\lambda$}
From now on, we will restrict our analysis to the case $p=P=0$, that is, we assume that the matter content of the Universe consists of pressureless dust.

Moreover, we consider that there is a linear relationship between the effective dark pressure and dark energy, which can be formulated generally as
\begin{equation}
P_{eff}(z)=-\sigma (z)r_{eff}(z)-\lambda,
\end{equation}%
where $\sigma (z)$, the parameter of the equation of state of the dark
energy, is assumed to be a function of the redshift, and $\lambda$ is a constant. For $\sigma (z)$ we
adopt the Chevalier-Polarski-Linder (CPL) parametrization \cite{Pol, Pol1}, with
\begin{equation}
\sigma (z)=\sigma _{0}+\sigma _{a}\frac{z}{1+z},
\end{equation}%
where $\sigma _{0}$ and $\sigma _{a}$ are constants. Then, the dynamical
system of equations describing the cosmological evolution takes the form
\begin{eqnarray}\label{M1}
-2(1+z)h(z)\frac{dh(z)}{dz}+3h^2(z)&=& 3\lambda +3\sigma (z)\nonumber\\
&&\times \left[2h(z)\Omega (z)-\Omega ^2(z)\right] ,\nonumber\\
\end{eqnarray}
and
\begin{eqnarray}\label{M2}
\hspace{-0.5cm}-2(1+z) h(z)\frac{d\Omega (z)}{dz}&=&2\left[ 3\sigma (z)-2\right] h(z)\Omega
(z)\nonumber\\
\hspace{-0.5cm}&&+\left[ 1-3\sigma (z)\right] \Omega ^{2}(z)+3\lambda ,
\end{eqnarray}
respectively.

The system of Eqs.~(\ref{M1}) and (\ref{M2}) must be solved with the initial
conditions $h(0)=1$, and $\Omega (0)=\Omega _{0}$. However, the initial
condition for $\Omega $ is determined by the first Friedmann equation (the
closure relation), Eq. (\ref{Fz1}) as $1=r(0)+2\Omega _{0}-\Omega _{0}^{2}$,
or $\left(\Omega _0-1\right)^2=r(0)$, giving
\begin{equation}
\Omega _{0}=1\pm \sqrt{r(0)}.
\end{equation}

Hence, the initial value of $\Omega $ is fully determined by the present day
value of the matter density. The variations as functions of the redshift of the Hubble function and of
the deceleration parameter are represented for Model II in Fig.~\ref{fig1}.

\begin{figure*}[htbp]
\centering
\includegraphics[width=0.490\linewidth]{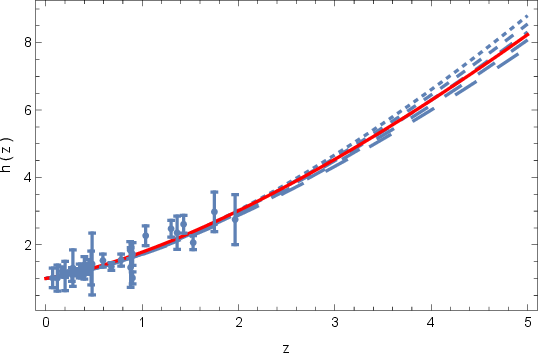} %
\includegraphics[width=0.490\linewidth]{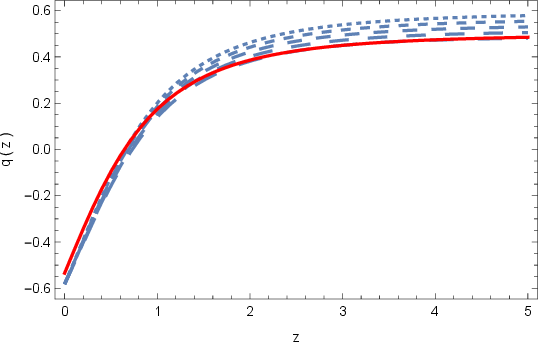}
\caption{Variations as a function of the redshift $z$ of the dimensionless
Hubble function (left panel), and of the deceleration parameter $q(z)$
(right panel) for Model II, for $\lambda =0.79$, $r(0)=0.311$, $%
\sigma _0=-0.10$, and different values of $\sigma _a$: $%
\sigma _a=0.04$ (dotted curve), $\sigma _a=0.06$ (short dashed
curve), $\sigma _a=0.08$ (dashed curve), $\sigma _a=0.10$ (long
dashed curve), and $\sigma _a=0.12$ (ultra-long dashed curve),
respectively. The predictions of the $\Lambda$CDM model are represented by
the red curve.}
\label{fig1}
\end{figure*}

As one can see from Fig.~\ref{fig1}, for the considered range of parameters,
Model II gives a good description of the observational data, and of the $%
\Lambda$CDM standard model up to a redshift of $z=5$. There are almost no deviations from $\Lambda$CDM in
the case of the deceleration parameter, with slightly higher values of $q$ predicted at higher redshifts. The behavior of the cosmological parameters is generally dependent on the adopted numerical values of the model parameters.

The variation with respect to the redfshift
of the matter density, and of the function $\Omega$ are represented in Fig.~%
\ref{fig2}.

\begin{figure*}[htbp]
\centering
\includegraphics[width=0.490\linewidth]{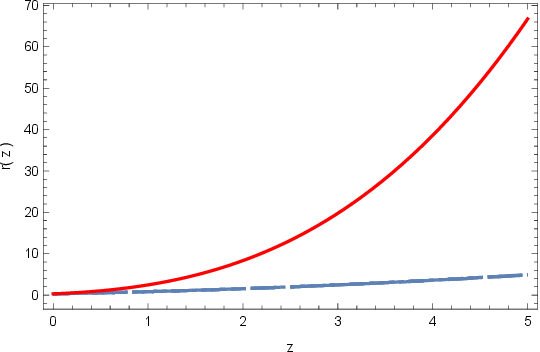} %
\includegraphics[width=0.490\linewidth]{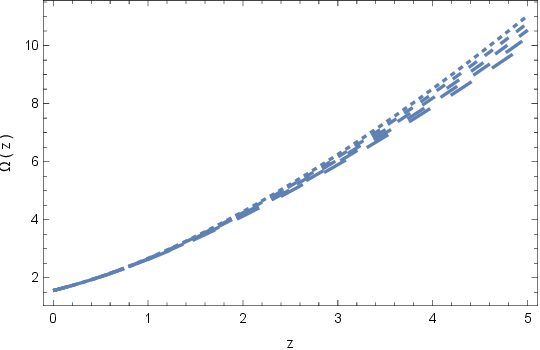}
\caption{Variations as a function of the redshift $z$ of the dimensionless
matter density $r(z)$ (left panel), and of the torsion vector component $
\Omega (z)$ (right panel) for Model II, for $\lambda =0.67$, $r(0)=0.311$, $
\sigma _0=-0.10$, and different values of $\sigma _a$:$
\sigma _a=0.04$ (dotted curve), $\sigma _a=0.06$ (short dashed
curve), $\sigma _a=0.08$ (dashed curve), $\sigma _a=0.10$ (long
dashed curve), and $\sigma _a=0.12$ (ultra-long dashed curve),
respectively. The predictions of the $\Lambda$CDM model are
represented by the red curve.}
\label{fig2}
\end{figure*}

The predictions of Model II for the behavior of the matter density do
coincide, up to $z\approx 1$, to the predictions of the $\Lambda $CDM model. However,
at higher redshifts, for the adopted range of parameters, Model II
predicts generally a much lower matter density than the standard $\Lambda$CDM paradigm.

The behavior of the function $Om(z)$ is represented in Fig.~\ref{fig3}. The $Om(z)$ diagnostic function has a very different behavior as compared to the diagnostic function of the standard $\Lambda$CDM model, indicating the possibility of the existence of several distinct cosmological regimes, and the presence at low redshifts ($z\approx 0.5$ of transitions from the quintessence to phantom like dynamical behaviors of the dark energy.

\begin{figure}[htbp]
\centering
\includegraphics[width=1.0\linewidth]{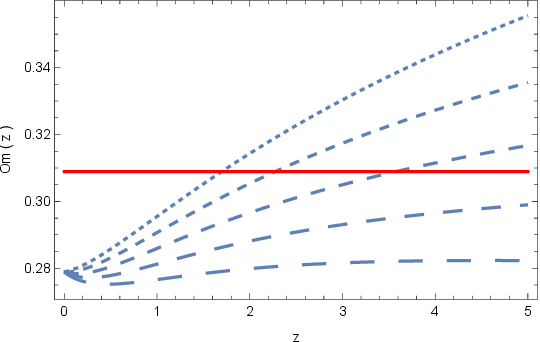}
\caption{Behavior of the function $Om (z)$ for Model II, for $\lambda =0.79$, $r(0)=0.311$, $
\sigma _0=-0.10$, and different values of $\sigma _a$: $%
\sigma _a=0.04$ (dotted curve), $\sigma _a=0.06$ (short dashed
curve), $\sigma _a=0.08$ (dashed curve), $\sigma _a=0.10$ (long
dashed curve), and $\sigma _a=0.12$ (ultra-long dashed curve),
respectively. The predictions of the $\Lambda$CDM model are
represented by the red curve.}
\label{fig3}
\end{figure}

\subsection{Model III: \;$P_{eff}=K r_{eff}^2$}

As a third cosmological model, we consider the case when the effective pressure and energy density satisfy
\begin{equation}
    P_{eff}=K r_{eff}^2,
\end{equation}
where $K={\rm constant}$, which means that we consider a polytropic equation of state with index $n=1$. In this case, we have
\begin{equation}
    -\frac{1}{3}  \left( 4h \Omega - \Omega^2 + 2 \frac{d \Omega}{d \tau} \right)=K \left( 2h \Omega - \Omega^2 \right)^2,
\end{equation}
or equivalently, in redshift variables
\begin{equation}
\begin{aligned}
    &-\frac{1}{3} \left[4 h(z) \Omega(z) - \Omega^2(z) -2(1+z)h(z) \frac{d \Omega}{dz}\right]\\
    &=K \left[ 2h(z) \Omega(z) - \Omega^2(z) \right]^2.
\end{aligned}
\end{equation}
Hence, the evolution equations for the polytropic dark energy model take the form
\begin{equation}
\begin{aligned}
    &-2(1+z)h(z) \frac{dh(z)}{dz}+3h^2(z)-4h(z) \Omega(z) + \Omega^2(z)\\
    &+2(1+z)h(z) \frac{d \Omega(z)}{dz}=0,
\end{aligned}
\end{equation}
\begin{equation}
\begin{aligned}
&\frac{1}{3} \left[4 h(z) \Omega(z) - \Omega^2(z) -2(1+z)h(z) \frac{d \Omega}{dz}\right]\\
&+K \left[ 2h(z) \Omega(z) - \Omega^2(z) \right]^2=0.
\end{aligned}
\end{equation}
Having solved the above system for initial conditions $h(0)=1, \Omega(0)=\Omega_0$, the energy density can be obtained from the closure relation
\begin{equation}
    r(z)=h^2(z)-2h(z) \Omega(z) + \Omega^2(z).
\end{equation}

The variations as functions of redshift of the Hubble function and of the
deceleration parameter are represented, for different values of $\Omega(0)$, in
Fig.~\ref{fig4}.

\begin{figure*}[htbp]
\centering
\includegraphics[width=0.490\linewidth]{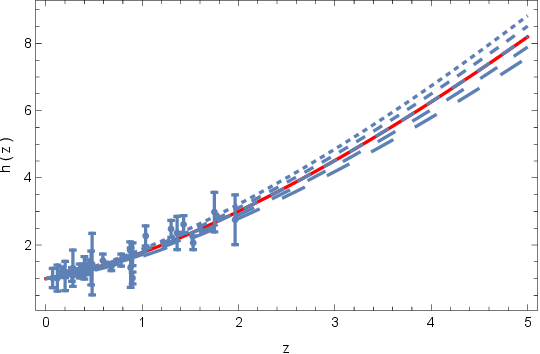} %
\includegraphics[width=0.490\linewidth]{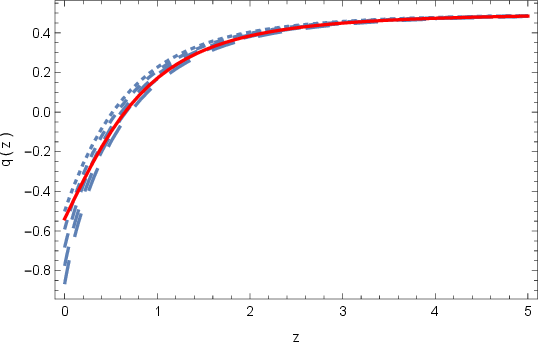}
\caption{Variations of the dimensionless Hubble function $h(z) $(left panel), and of the deceleration parameter $q(z)$ (right panel) for Model III with $K=-2$ and initial conditions $\Omega(0)=0.35$ (dotted curve), $\Omega(0)=0.37$ (short dashed curve), $\Omega(0)=0.39$ (dashed curve) , $\Omega(0)=0.41$ (long dashed curve),
$\Omega(0)=0.43$ (ultra-long dashed curve),respectively. The observational data for the Hubble function are represented with their error bars, while the red curve depicts the predictions of the $\Lambda$CDM model.}\label{fig4}
\end{figure*}

As one can see from Fig.~\ref{fig4}, Model III gives a very good description
of the observational data, as well as of the $\Lambda$CDM model. However, some small differences between the
evolution of the deceleration parameter do appear in the semi-symmetric metric gravity theory, as compared to the $\Lambda$CDM model, at low redshifts. The variations of the dimensionless matter density $r (z)$ and of the dimensionless torsion vector $\Omega (z)$ are
represented in Fig.~\ref{fig5}.

\begin{figure*}[htbp]
\centering
\includegraphics[width=0.490\linewidth]{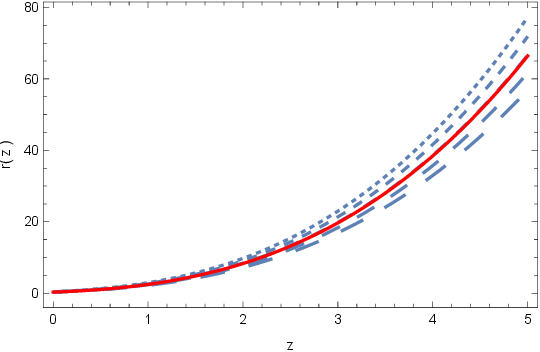} %
\includegraphics[width=0.490\linewidth]{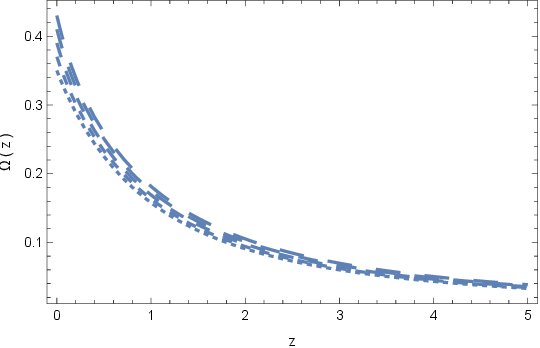}
\caption{Variations of the dimensionless matter energy density $r(z)$ (left panel), and of the torsion vector $\Omega(z)$ (right panel) for Model III with $K=-2$ and initial conditions $\Omega(0)=0.35$ (dotted curve), $\Omega(0)=0.37$ (short dashed curve), $\Omega(0)=0.39$ (dashed curve) , $\Omega(0)=0.41$ (long dashed curve),
$\Omega(0)=0.43$ (ultra-long dashed curve),respectively. The red curve represents the predictions of the $\Lambda$CDM model.}\label{fig5}
\end{figure*}

 As can be seen from Fig \ref{fig5}, the present day matter density predicted by Model III basically coincide with the predictions of the $\Lambda$CDM model.  $\Omega (z)$ monotonically decreases as a function of redshift, and it takes only positive values. The $Om(z)$ diagnostic function,  plotted in Fig.~\ref{fig6},  is monotonically increasing with the redshift, and it indicates the existence of significant differences between the present model and $\Lambda$CDM, which are more important at lower redshifts.

\begin{figure}[htbp]
\centering
\includegraphics[width=1.0\linewidth]{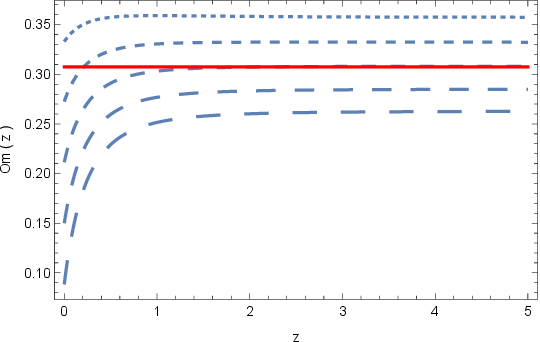}
\caption{Behavior of the function $Om z$ for Model III with $K=-2$ and initial conditions $\Omega(0)=0.35$ (dotted curve), $\Omega(0)=0.37$ (short dashed curve), $\Omega(0)=0.39$ (dashed curve), $\Omega(0)=0.41$ (long dashed curve),
$\Omega(0)=0.43$ (ultra-long dashed curve), respectively. The red curve represents the predictions of the $\Lambda$CDM model.}\label{fig6}
\end{figure}

\section{Discussion and final remarks}\label{section6}

The generalizations of the Riemannian geometry have for a long time attracted the attention of both mathematicians and physicists. In particular, geometries with torsion and nonmetricity have offered a very fruitful background for the development of physical theories that could offer a description of the gravitational interaction, going beyond general relativity. A large variety of such theories have been constructed, and their physical role was widely investigated. A very interesting result is related to the so-called geometrical trinity of gravity \cite{G3}, which is based on the unexpected result that the gravitational force can be described equivalently by using three independent formulations, based on curvature, torsion or nonmetricity, respectively. Each of these descriptions generate their own class of theories, which can be successfully used for the description of the gravitational phenomenology, and especially for the observational aspects of cosmology. General relativity, a gravitational theory based solely on the Riemannian curvature, still faces the problem of the cosmological constant, when confronted with observations, despite the excellent description of the data by the $\Lambda$CDM model. Theories based on torsion or nonmetricity have, on the other hand, a firm geometrical foundation, even that the physical interpretation of torsion and nonmetricity may still raise some unsolved questions. To explain the observational data, the standard $\Lambda$CDM paradigm postulates the existence of the dark energy and of the cold dark matter. But even after several decades of experimental search, no conclusive evidence has been found for the existence of such dark matter particles. Dark energy remains also elusive from the observational, or even experimental point of view. Thus the situation of the present day gravitational physics opens a new window for the consideration of extended gravitational theories, which could give an alternative explanation of the existing observational data, and of the lack of experimental data.

In the present work we have investigated a torsional extension of general relativity, based on a special type of torsion.  Interestingly enough, this notable extension of Riemannian geometry was initially introduced by Friedmann and Schouten in 1924 \cite{FriedmannSchouten}, through the notion of a semi-symmetric connection. This work was published by Friedmann around two years after his fundamental paper on relativistic cosmology did appear. If the mathematical landscape of semi-symmetric metric connections is relatively well-explored, their physical implications remain largely ignored. The semi-symmetric metric connections can also be related to another interesting connection, the Schr\"{o}dinger connection, whose physical implications are also unexplored to a great extent. Hence, the objective of this paper is twofold: firstly, we have investigated the relation between the semi-symmetric metric connection and the Schr\"{o}dinger connection; and secondly, we have examined the  extent to which semi-symmetric metric gravity is a viable extension of general relativity.

The basic mathematical characteristic of the semi-symmetric metric gravity is given by the specific, and very simple form of the torsion tensor, which is represented as $\tensor{T}{^\mu_\nu _\rho}=\pi_{\rho} \delta^{\mu}_{\nu} - \pi_{\nu} \delta^{\mu}_{\rho}$, where $\pi _\rho$ is the torsion vector, representing an arbitrary four-vector describing the geometric characteristic of the spacetime manifold. In the absence on nonmetricity, the mathematical expressions of the basic geometric quantities are rather simple (as compared to other modified gravitational theories), and one could easily construct the geometric Einstein tensor. As a first approximation to gravitational phenomena we postulate that, similarly to standard general relativity, the Einstein tensor is proportional to the energy-momentum tensor of matter. This allows to formulate the gravitational field equations of the semi-symmetric metric gravity, as given by Eqs.~(\ref{Einsteinsemisymmetricequation}).

As a next step in our investigations we have considered the cosmological implications of the theory. In an isotropic and homogeneous geometry, the torsion vector has only a temporal component $\omega (t)$. The two generalized Friedmann equations can be obtained in a straightforward way, and they contain some extra terms that we interpret as representing a dynamical dark energy. The effective dark energy depends only on $\omega$, and the Hubble function $H$, while the effective dark pressure also contains the time derivative of the torsion vector. The system of the generalized Friedmann equations also admits a vacuum de Sitter type solution, with the torsion tensor becoming a constant in the large time limit. de Sitter type solutions do also exist in the presence of matter.

In order to close the system of cosmological equations we need a supplementary evolution equation for the torsion tensor. In the present formalism of the semi-symmetric metric gravity such an equation cannot be obtained independently, and hence we need to resort to some physical and mathematical assumptions.

In this work we have considered three independent cosmological models. In the first cosmological model we have assumed that the effective dark energy and pressure are constants. This assumption allows us to obtain a full analytical solution of the field equations. The expressions of the Hubble function and of the scale factor, as well as the matter energy density and pressure can be obtained in a simple form. The model has a unique de Sitter limit, independent on the model parameters. However, there is no physical equation of state relating the matter density and pressure, even that an equation of state given in parametric form can be considered as describing the thermodynamic properties of matter.

In the second model we have imposed an equation of state for the effective pressure of the geometric dark energy. We have assumed a simple linear expression, with a redshift dependent parameter of the equation of state. The resulting model was compared with the observational data, and with the predictions of the standard $\Lambda$CDM model. This semi-symmetric metric cosmological model can give a good description of the observational data, and reproduce the predictions of the $\Lambda$CDM model both quantitatively, and qualitatively, in the case of the behavior of the matter energy density, where the model predicts slightly lower values at high redshifts.

However, some differences do appear in the behavior of the $Om(z)$ diagnostic function. While in $\Lambda$CDM, $Om(z)$ is a constant, in the semi-symmetric gravity the $Om(z)$ function has a complicated behavior, indicating the existence of a possible transition from a quintessence like behavior in the redshift range  $z\in (0,0.5)$ to a phantom-like behavior for $z>0.5$. Other similar transitions are also possible at higher redshifts. The cosmological behavior of this model is strongly influenced by the numerical values of the CPL parametrization of the effective equation of state, as well on the present day value of the matter density, which determines the present day value of the torsion vector through the relation $\left[\Omega (0)-1\right]^2=r(0)$, which follows directly from the first Friedmann equation. Hence, in this model there is an important relation between torsion and matter, with either matter determining (generating?) torsion, or with the inverse process taking place. During the cosmological evolution, the torsion tensor takes only positive values.

In the third cosmological model we have considered we have imposed a polytropic equation of state to describe the relation between the dark energy effective pressure and density, of the form $P_{eff}\sim r_{eff}^2$.  The cosmological evolution equations can again be formulated as a system of two strongly nonlinear ordinary differential equations, whose solutions can be obtained only numerically. It is important to note that in these models the cosmological evolution is determined by the initial value of $\Omega(0)$ only, for a fixed present day value of the normalized Hubble function. The polytropic semi-symmetric metric gravity model also gives a good description of the observational data for the Hubble function, and it can reproduce almost exactly the predictions of the $\Lambda$CDM model. The ordinary matter energy density also reproduces well the $\Lambda$CDM predictions. The torsion vector, a monotonically decreasing function of redshift,  has only positive values in the considered range of $z$.

The $Om(z)$ diagnostic function has, in Model III, a very different behavior as compared to Model I, indicating that the dark energy is quintessence-like during the entire cosmological evolution, and there is no transition to a phantom-like phase.

We would like to point out that our cosmological results are mostly qualitative in nature. In order to confirm/disapprove the validity of the semi-symmetric metric gravity theory cosmology, a detailed statistical analysis, requiring the investigation of a large number of datasets, is necessary.

From a cosmological perspective, the matter energy density is generally not conserved in the semi-symmetric metric gravity theory. However, this raises the question of the physical interpretation of the possible non-conservation effects within the theory.

One such possibility is to assume that the nonconservation of the matter energy-momentum tensor is related to particle creation effects. This type of interpretation was also considered in the framework of the modified gravity theories with geometry-matter coupling \cite{fRLm,fRT,book}. On the other hand, the particle creation processes are also a consequence of the quantum field theories in curved space-times, as first shown in \cite{Parker,Parker1,Parker2}, and they are a direct consequence of the time
evolution of the gravitational field.

Therefore, the present theory, once it is interpreted physically as describing particle creation in a cosmological background, could also lead to the possibility of an effective semiclassical approach for the description of the quantum field theoretical processes in time-dependent gravitational fields.

From the classical point of view matter creation can be described by using the formalism of the irreversible thermodynamics of open systems \cite%
{P-M,Lima,Su, Bar}. In the presence of matter generation processes, since the covariant divergences of the fundamental equilibrium thermodynamic quantities are different from zero,  the equilibrium equations must be adjusted to incorporate particle creation \cite{P-M,Lima,Su}.

The particle flux $N^{\mu} \equiv nu^{\mu}$, where $n$ is the particle number density, is thus described in the presence of gravitational particle creation by the balance equation
$
\nabla _{\mu}N^{\mu}=\dot{n}+3Hn=n\Psi,
$
where by $\Psi $ is the particle generation rate. If the condition $\Psi \ll H $ is satisfied, the particle creation processes are negligibly small as compared to the expansion rate of the Universe.

Another fundamental physical quantity, the entropy flux vector $S^\mu$ is defined as $%
S^{\mu} \equiv su^{\mu} = n\sigma u^{\mu}$, where $s$ is the entropy density, and $\sigma $ is the entropy per particle. According to the second law of thermodynamics, the
divergence of the entropy flux must satisfy the condition
$
\nabla _{\mu}S^{\mu}=n\dot{\sigma}+n\sigma \Psi\geq 0.
$
If the entropy density $\sigma $ is a constant, then obviously
$
\nabla _{\mu}S^{\mu}=n\sigma \Psi =s\Psi\geq 0.
$

Thus, in the case  $\sigma =\mathrm{constant}$, the variation of the total entropy is exclusively due to the gravitational, adiabatic matter creation processes. Since by definition $s>0$, the particle creation rate $\Psi$ must satisfy the condition $\Psi \geq 0$, which  shows that gravitational fields can create matter, but the inverse
process cannot take place.

In the presence of particle creation, the energy-momentum tensor of a
physical system must also be corrected to include particle creation, as well
as the second law of thermodynamics, so that it takes the form
$
T^{\mu \nu}=T^{\mu \nu}_\text{eq}+\Delta T^{\mu \nu},
$ \cite{Bar}
where $T^{\mu \nu}_\text{eq}$ represents the equilibrium component, while $%
\Delta T^{\mu \nu}$ is the modification due to the presence of matter
creation.

Hence, one can write generally
$
\Delta T_{\; 0}^0=0, \quad \Delta T_{\; i}^j=-p_c\delta_{\; i}^j,
$
where $p_c$ is the creation pressure, which represents, in a
phenomenological description, the effects of matter generation in a
macroscopic physical system. The
contribution of particle creation to the matter energy-momentum tensor is thus given by
$
\Delta T^{\mu \nu}=-p_ch^{\mu \nu}=-p_c\left(g^{\mu
\nu}+u^{\mu}u^{\nu}\right)
$ \cite{Bar},
giving
$
u_{\mu}\nabla _{\nu}\Delta T^{\mu \nu}=3Hp_c.
$

The total thermodynamic energy
balance equation, $u_{\mu}\nabla _{\nu}T^{\mu \nu}=0$, gives the
generalized energy conservation equation in the presence of particle creation
\be
\dot{\rho}+3H\left(\rho+p+p_c\right)=0.
\ee

The Gibbs law must also be satisfied by the thermodynamic quantities, and,
in the presence of matter creation, it can be written down as \cite{Lima}
\begin{equation}
n T d \left(\frac{s}{n}\right)=nTd%
\sigma=d\rho -\frac{\rho+p}{n}dn,
\end{equation}
where by $T$ we have denoted the thermodynamic temperature of the
given system.

After some simple and straightforward algebraic manipulations, the energy balance equation~(\ref{eqcons}) can be rewritten as
\begin{equation}
\dot{\rho}+3H\left( \rho +p+p_{c}\right) =0,  \label{76}
\end{equation}%
with the creation pressure $p_{c}$ of semi-symmetric metric gravity defined as
\begin{eqnarray}\label{pc1}
p_{c} =\frac{\omega }{8\pi }\Bigg[2\frac{\dot{H}}{H}-\frac{2\dot{\omega}}{H}+2H-2\omega\Bigg].
\end{eqnarray}%

Then, with the use of the creation pressure, the generalized energy balance
equation~(\ref{eqcons}) can be obtained, in a way similar to standard
general relativistic cosmology, from the vanishing of the divergence of the
total energy momentum tensor $T^{\mu \nu }$, defined as
$
T^{\mu \nu }=\left( \rho +p+p_{c}\right) u^{\mu }u^{\nu }+\left(
p+p_{c}\right) g^{\mu \nu },
$
and with $p_c$ given by Eq.~(\ref{pc1}). To obtain the conservation equation, one must also adopt the comoving frame
for the cosmological expansion.

By assuming adiabatic particle production, with $\dot{\sigma}=0$, from the
Gibbs law we obtain
$
\dot{\rho} =\left(\rho+p\right)\dot{n}/n =\left(\rho+p\right)\left(%
\Psi-3H\right).
$
By using the energy balance equation (\ref{76}) we obtain the relation
between the creation pressure and the particle creation rate as
\begin{equation}
\Psi=-3H\frac{p_c}{\rho+p}.
\end{equation}

Thus, in the semi-symmetric metric gravity theory the particle creation rate is given by the general
expression
\begin{eqnarray}
\hspace{-0.5cm}\Psi  &=&-\frac{3H\omega }{8\pi \left( \rho +p\right) }\Bigg[2\frac{\dot{H}}{%
H}-\frac{2\dot{\omega}}{H}+2H-2\omega\Bigg].
\end{eqnarray}

The particle creation rate $\Psi $ must satisfy the condition $\Psi \geq 0$,
which, by taking into account that $H$, $\rho $, and $p$ are all positive,
is equivalent to the condition $p_{c}<0$ for all times. Hence, the condition
of the negativity of the creation pressure imposes a strong constraint on
the physical parameters of the  semi-symmetric gravity theory.

By taking into account that $2\dot{H}/H=-2(1+q)H$, where $q$ is the
deceleration parameter, the particle creation rate takes the form
\begin{eqnarray}
\hspace{-0.8cm}\Psi  =\frac{3H\omega }{8\pi \left( \rho +p\right) }\Bigg[2qH+\frac{2\dot{\omega}}{H}+2\omega  \Bigg].
\end{eqnarray}

The divergence of the entropy flux vector is obtained in terms of the
creation pressure as
$
\nabla _{\mu}S^{\mu}=-3 n \sigma H p_c/(\rho +p).
$
The condition $p_{c}<0$ assures the positivity for all times of the entropy production
rate, as required by the second law of thermodynamics. Explicitly, the
entropy production rate in the  semi-symmetric gravity theory is obtained as
\be
\nabla _{\mu }S^{\mu } =\frac{3n\sigma H\omega }{8\pi (\rho +p)}\Bigg[%
2qH+\frac{2\dot{\omega}}{H}+2\omega\Bigg].
\ee

We consider now the temperature evolution of the newly created particles.  We assume that the fluid is
described thermodynamically by two equations of state for the density and
pressure, which are given in the general form $\rho =\rho \left( n,T\right) ,p=p\left( n,T\right)$. Then, the temperature evolution of the newly
created particle in a relativistic fluid can be obtained as \cite{Su}
\begin{equation}
\frac{\dot{T}}{T}=\left( \frac{\partial p}{\partial \rho
}\right) _{n}\frac{\dot{n}}{n}=c_{s}^{2}\frac{\dot{n}}{n},  \label{102}
\end{equation}%
where $c_{s}^{2}=\left( \partial p/\partial \rho \right) _{n}$ is the speed
of sound. Eq.~(\ref{102}) can also be rewritten as
\begin{equation}
\frac{\dot{T}}{T}=c_{s}^{2}\left( \Psi -3H\right)
=-3c_{s}^{2}H\left( 1+\frac{p_{c}}{\rho +p}\right) .
\end{equation}%

Hence, in the semi-symmetric metric gravity theory, the time variation of the temperature of the newly created
particles is given by
\begin{eqnarray}
\hspace{-0.5cm}\frac{\dot{T}}{T} =3c_{s}^{2}H\Bigg\{1+\frac{\omega }{%
8\pi (\rho +p)}\Bigg[2qH+\frac{2\dot{\omega}}{H}+2\omega  \Bigg]\Bigg\}.
\end{eqnarray}%

For the temperature of the particles to increase, $\dot{T}%
>0$, the thermodynamic condition $1+p_{c}/(\rho +p)<0$ must be satisfied,
which is equivalent, for $\rho >0$, $p>0$, to $\rho +p+p_{c}<0$, or $\rho
+p<-p_{c}$.
If $\left(\partial p/\partial \rho\right)_n=c_s^2=\gamma =\mathrm{\ constant}%
>0$, we find the temperature-newly created particle number relation as given
by the simple power law expression $T \sim n^\gamma$.

Thus, we have shown that the semi-symmetric gravity theory admits a full thermodynamic interpretation in terms of the thermodynamics of open systems. This interpretation may lead to some new insights into the physical foundations of the theory, and of its possible relation with some phenomenological models of quantum gravity, which also predict the existence of particle creation due to the time evolution of the gravitational field \cite{Parker, Parker1,Parker2}. We would also like to point out that in the open irreversible thermodynamic model of particle creation of the semi-symmetric metric gravity, the creation of the particles can also take place in the form of dynamical scalar (perhaps Higgs) fields $\sigma$, originating from the fields $\omega$ and $\dot{\omega}$, according to a reaction of the form $\omega +\omega \rightarrow \sigma +\sigma$. However, the nature of the particle produced due to the decay of the torsion vector cannot be predicted by the present theory, but we can still assume that they    decay into dark or ordinary matter. Thus, for example, dark matter in the form of a scalar field may  be a consequence of the particle creation due to the torsion decay taking place in the early phases of the evolution of the Universe. A similar creation of dark energy particles from torsion cannot be either excluded.

To conclude our study, in the present paper we have shown  that the extended gravitational theory obtained in the mathematical formalism of the semi-symmetric metric gravity theory can fully explain the dynamical evolution of the Universe, its accelerated expansion, a possible early inflationary behavior,
and, perhaps, even the formation of cosmic structures in the Universe. The present model does not introduce/postulates any new physical or coupling parameters in its formulation, and a single geometric quantity, the torsion, describes the behavior of Universe at all cosmological scales. Several cosmological models can be obtained in this framework, generally non-conservative with respect to the matter energy-density. These models can describe the cosmological data under a number of realistic assumptions, without the need of introducing a cosmological constant. Still, these cosmological models automatically generate a dark energy (and perhaps even dark matter) component. Unlike general relativity, the semi-symmetric metric gravity theory can also provide a phenomenological description of particle creation, which could prove to be of fundamental importance in the understanding of the very early stages of the Universe, and of the processes that led to the emergence of matter. Thus, the results presented in this paper can open some new perspectives for the comprehension of the complex aspects of the cosmic dynamics, and of the role of the
torsion, and other geometric quantities, in an eternally evolving Universe.

 \section*{Acknowledgments}

We would like to thank to Prof. Carlo Alberto Mantica for pointing out an important sign error in the first version of our manuscript. The work of T.H. is supported by a grant from the
Romanian Ministry of Education and Research, CNCS-UEFISCDI, project number PN-III-P4-ID-PCE2020-
2255 (PNCDI III). The work of L.Cs. is supported by Collegium Talentum  Hungary and the StarUBB research fellowship.

\begin{appendix}

\section{Decomposition of the affine connection}\label{appendixA}

This Section is devoted to the proof of the formula \eqref{generalconnection}. We start by using the definition of non-metricity in our convention. Doing a circular permutation results in
\begin{equation}\label{nonmetricity1}
 Q_{\lambda \nu \rho}=-\nabla_{\lambda} g_{\nu \rho}=-\partial_{\lambda} g_{\nu \rho} + \tensor{\Gamma}{^\beta _\nu _\lambda} g_{\beta \rho} + \tensor{\Gamma}{^\beta _\rho _\lambda} g_{\nu \beta} ,
\end{equation}
\begin{equation}\label{nonmetricity2}
Q_{\rho \lambda \nu}=- \nabla_{\rho} g_{\lambda \nu}= - \partial_{\rho} g_{\lambda \nu} + \tensor{\Gamma}{^\beta_\lambda _\rho} g_{\beta \nu} + \tensor{\Gamma}{^\beta _\nu _\rho} g_{\lambda \beta},
\end{equation}
\begin{equation}\label{nonmetricity3}
Q_{\nu \rho \lambda}=-\nabla_{\nu} g_{\rho \lambda}=-\partial_{\nu} g_{\rho \lambda} + \tensor{\Gamma}{^\beta _\rho _\nu}g_{\beta \lambda} + \tensor{\Gamma}{^\beta _\lambda _\nu} g_{\rho \beta}.
\end{equation}

We add \eqref{nonmetricity3}, \eqref{nonmetricity2} and subtract \eqref{nonmetricity1}, which results in
\bea
 Q_{\nu \rho \lambda} +Q_{\rho \lambda \nu} -Q_{\lambda \nu \rho}&=&- \left(\partial_{\nu} g_{\rho \lambda} +\partial_{\rho} g_{\lambda \nu}  -\partial_{\lambda} g_{\nu \rho}\right)\nonumber\\
    &+&g_{\rho \beta} \left(\tensor{\Gamma}{^\beta _\lambda _\nu}- \tensor{\Gamma}{^\beta _\nu _\lambda}\right)\nonumber\\
    &+& g_{\beta \lambda} \left(\tensor{\Gamma}{^\beta _\rho _\nu}+ \tensor{\Gamma}{^\beta _\nu _\rho} \right)\nonumber\\
    &+& g_{\beta \nu} \left(\tensor{\Gamma}{^\beta_\lambda _\rho} - \tensor{\Gamma}{^\beta _\rho _\lambda} \right).
\eea

Introducing the symmetrization and antisymmetrization notation gives
\bea
    Q_{\nu \rho \lambda} +Q_{\rho \lambda \nu} -Q_{\lambda \nu \rho}&=&- \left(\partial_{\nu} g_{\rho \lambda} +\partial_{\rho} g_{\lambda \nu}  -\partial_{\lambda} g_{\nu \rho}\right)\nonumber\\
    &+& 2 g_{\rho \beta} \tensor{\Gamma}{^\beta_{[\lambda \nu]}} \nonumber\\
    &+& 2 g_{\beta \lambda} \tensor{\Gamma}{^\beta_{(\nu \rho)}}\nonumber\\
    &+&2g_{\beta \nu} \tensor{\Gamma}{^\beta_{[\lambda \rho]}}.
\eea

We use our convention, namely that torsion is twice the antisymmetric part of the Christoffel symbols and that the symmetric part can be rewritten as the whole part minus the antisymmetric part to obtain
\bea
    Q_{\nu \rho \lambda} +Q_{\rho \lambda \nu} -Q_{\lambda \nu \rho}&=&- \left(\partial_{\nu} g_{\rho \lambda} +\partial_{\rho} g_{\lambda \nu}  -\partial_{\lambda} g_{\nu \rho}\right)\nonumber\\
    &+&g_{\rho \beta} \tensor{T}{^\beta _\nu _\lambda}\nonumber\\
    &+&2g_{\beta \lambda}  \tensor{\Gamma}{^\beta_{\nu \rho}} - g_{\beta \lambda} \tensor{T}{^\beta}_{\rho \nu}\nonumber\\
    &+&g_{\beta \nu} \tensor{T}{^\beta _\rho _\lambda}.
\eea
We lower now the indices with the help of the metric tensor, and thus
\begin{equation}
    \begin{aligned}
    Q_{\nu \rho \lambda} +Q_{\rho \lambda \nu} -Q_{\lambda \nu \rho}=&- \left(\partial_{\nu} g_{\rho \lambda} +\partial_{\rho} g_{\lambda \nu}  -\partial_{\lambda} g_{\nu \rho}\right)\\
    &+T_{\rho \nu \lambda} \\
    &+2g_{\beta \lambda} \tensor{\Gamma}{^\beta_{\nu \rho}} - T_{\lambda \rho \nu}\\
    &+ T_{\nu \rho \lambda}.
    \end{aligned}
\end{equation}
Multiplying by $g^{\mu \lambda}$ yields
\begin{equation}
    \begin{aligned}
        g^{\mu \lambda} \left(Q_{\nu \rho \lambda} +Q_{\rho \lambda \nu} -Q_{\lambda \nu \rho} \right)=&-g^{\mu \lambda} \left(\partial_{\nu} g_{\rho \lambda} +\partial_{\rho} g_{\lambda \nu}  -\partial_{\lambda} g_{\nu \rho}\right)\\
        &+g^{\mu \lambda} \left( T_{\rho \nu \lambda}- T_{\lambda \rho \nu} + T_{\nu \rho \lambda} \right)\\
        &+2 \delta^{\mu}_{\beta} \tensor{\Gamma}{^\beta _\nu _\rho}.
    \end{aligned}
\end{equation}
Via a simple algebraic manipulation, we express the Christoffel symbols from the right hand side:
\begin{equation}
    \begin{aligned}
        \tensor{\Gamma}{^\mu _\nu _\rho}&=\frac{1}{2} g^{\mu \lambda} \left(\partial_{\nu} g_{\rho \lambda} +\partial_{\rho} g_{\lambda \nu}  -\partial_{\lambda} g_{\nu \rho}\right)\\
        &+\frac{1}{2}g^{\mu \lambda} \left(Q_{\nu \rho \lambda} +Q_{\rho \lambda \nu} -Q_{\lambda \nu \rho} \right)\\
        &-\frac{1}{2} g^{\mu \lambda}(T_{\rho \nu \lambda}-T_{\lambda \rho \nu} + T_{\nu \rho \lambda}).
    \end{aligned}
\end{equation}
Identifying the Christoffels of the Levi-civita connection
\begin{equation}
    \tensor{\gamma}{^\mu _\nu _\rho}=\frac{1}{2} g^{\mu \lambda} \left(\partial_{\nu} g_{\rho \lambda} +\partial_{\rho} g_{\lambda \nu}  -\partial_{\lambda} g_{\nu \rho}\right)
\end{equation}
leads to the desired result
\begin{equation}
    \begin{aligned}
        \tensor{\Gamma}{^\mu _\nu _\rho}&=  \tensor{\gamma}{^\mu _\nu _\rho}  +\frac{1}{2}g^{\mu \lambda} \left(Q_{\nu \rho \lambda} +Q_{\rho \lambda \nu} -Q_{\lambda \nu \rho} \right) \\
         &-\frac{1}{2} g^{\mu \lambda}(T_{\rho \nu \lambda}-T_{\lambda \rho \nu} + T_{\nu \rho \lambda}).
    \end{aligned}
\end{equation}

\section{Derivation of the Friedmann Equations}\label{appendixC}

Here we derive the generalized Friedmann equations in semi-symmetric metric gravity theory, given the assumptions in Section \ref{Friedmannsection}. First, recall that the non-zero components of the Ricci tensor of the Levi-Civita connection are given by
\begin{equation}
    \overset{\circ}{R}_{00}=-3\frac{\ddot a}{a}, \; \; \overset{\circ}{R}_{11}=\overset{\circ}{R}_{22}=\overset{\circ}{R}_{33}=a \ddot a +2 \dot a^2.
\end{equation}
Similarly, the Ricci scalar takes the well known form
\begin{equation}
    \overset{\circ}{R}=6 \left(\frac{ \ddot a}{a}+\frac{\dot a ^2}{a^2} \right).
\end{equation}
The non-vanishing Christoffel symbols are given by
\begin{equation}
    \tensor{\gamma}{^0_i_j}=a \dot{a} \delta_{ij}, \; \; i,j=1,2,3;
\end{equation}
\begin{equation}
    \tensor{\gamma}{^i_0_j}=\frac{\dot a}{a} \delta^{i}_{j}, \; \; i,j=1,2,3.
\end{equation}
As a first step, we will write down the Einstein equation \eqref{Einsteinsemisymmetricequation} in the $00$ component, which reads
\begin{equation}\label{einstein00}
\begin{aligned}
    \overset{\circ}{R}_{00}&-\frac{1}{2} g_{00} \overset{\circ}{R} - \overset{\circ}{\nabla}_{0} \pi_{0} - \overset{\circ}{\nabla}_{0} \pi_{0} \\
    &+2 \pi_0 \pi_0 +2 g_{00} \overset{\circ}{\nabla}_{\lambda} \pi^{\lambda} +g_{00} \pi^{\rho} \pi_{\rho}=8 \pi T_{00}.
\end{aligned}
\end{equation}
From the conventions in Section \ref{Friedmannsection}, we have that
\begin{equation}
    \pi_{0}=-\omega, \; \; \pi^{0}=\omega, \; \; T_{00}=\rho
\end{equation}
which immediately implies
\begin{equation}
    \overset{\circ}{\nabla}_{0} \pi_{0}=\partial_{0} \pi_0=-\dot \omega.
\end{equation}
The covariant divergence of $\pi$ can be computed as
\begin{equation}
    \overset{\circ}{\nabla}_{\beta} \pi^{\beta}=\partial_{\beta} \pi^{\beta} + \tensor{\gamma}{^\beta _\beta _\rho} \pi^{\rho}=\dot{\omega} +3 \frac{ \dot {a}}{a}\omega.
\end{equation}
 Putting everything together, and substituting back into \eqref{einstein00} results in
\begin{equation}
\begin{aligned}
    -3 \frac{\ddot a}{a} &+ \frac{1}{2} 6 \left(\frac{ \ddot a}{a}+\frac{\dot a ^2}{a^2} \right)  + \dot \omega + \dot \omega \\
    &+2 \omega^2 -2 \left(\dot \omega +3 \frac{\dot a}{a} \omega \right) -(-\omega^2)=8 \pi p.
\end{aligned}
\end{equation}
It is readily seen that the terms containing second derivatives of time vanish, and similarly $\dot \omega$ vanishes. An algebraic simplification yields
\begin{equation}
    3 \frac{\dot a ^2}{a^2} - 6 \frac{\dot a}{a}\omega+3 \omega^2=8 \pi \rho.
\end{equation}
Introducing the Hubble parameter $H=\frac{\dot a}{a}$, the first Friedmann equation \eqref{F1}
\begin{equation}
    3H^2= 8 \pi \rho +6 H \omega - 3 \omega^2
\end{equation}
is obtained.

For the second Friedmann equation, we have to consider the $ii$ components of the Einstein equation \eqref{Einsteinsemisymmetricequation}, which are given by
\begin{equation}\label{einstein11}
\begin{aligned}
   \overset{\circ}{R}_{ii}&-\frac{1}{2} g_{ii} \overset{\circ}{R} - \overset{\circ}{\nabla}_i \pi_i - \overset{\circ}{\nabla}_i \pi_i \\
   &+2 \pi_i \pi_i + 2 g_{ii} \overset{\circ}{\nabla}_{\lambda} \pi^{\lambda} +g_{ii} \pi^{\rho} \pi_{\rho}= 8 \pi T_{ii}.
    \end{aligned}
\end{equation}
It can be easily seen that it does not matter whether $i=1,i=2,i=3$, the equations will be the same. For simplicity, we consider $i=1$, in which case a lot of terms vanish. This follows from our conventions
\begin{equation}
    \pi_{1}=0, \; \; T_{11}=pa^2.
\end{equation}
Similarly, we have
\begin{eqnarray}
    \nabla_i \pi_i=\partial_i \pi_i - \tensor{\Gamma}{^\rho_{ii}} \pi_{\rho}=0-\tensor{\Gamma}{^0_i_i} (-\omega)=a \dot a \omega.
\end{eqnarray}
Putting everything together, and substituting into \eqref{einstein11} leads to
\begin{equation}
\begin{aligned}
    a \ddot a &+2 \dot a^{2} -\frac{1}{2} a^2 6 \left( \frac{\ddot a}{a} +\frac{\dot a^2}{a^2} \right) - a \dot a \omega - a \dot a \omega \\
    &+0 +2 a^2 \left(\dot \omega + 3 \frac{\dot a}{a} \omega \right)+ a^2 \left( - \omega^2 \right)=8 \pi p a^2.
\end{aligned}
\end{equation}
Expanding the brackets yields
\begin{equation}
\begin{aligned}
    \textcolor{blue}{a \ddot{a}} &+ \textcolor{red}{2 \dot a^2} - \textcolor{blue}{3a \ddot a} - \textcolor{red}{3 \dot a^2}  -\textcolor{orange}{2 a \dot a \omega} \\
    &+2 a^2 \dot \omega +\textcolor{orange}{6 a \dot a \omega} - a^2 \omega^2=8\pi p a^2.
\end{aligned}
\end{equation}
Collecting the terms together gives
\begin{equation}
    -\textcolor{blue}{2 a \ddot a} - \textcolor{red}{\dot{a}^2}+ 2 a^2 \dot \omega+\textcolor{orange}{4a \dot a \omega} -a^2 \omega^2=8 \pi p a^2.
\end{equation}
Dividing by $a^2$ results in
\begin{equation}\label{friedmannwithas}
    -2\frac{\ddot a}{a} - \frac{\dot a^2}{a^2}+ 2\dot \omega +4 \frac{\dot a}{a} \omega - \omega^2=8 \pi p.
\end{equation}
By introducing $H=\frac{ \dot a}{a}$, it follows that
\begin{equation}
    \dot{H}=\frac{\ddot a a - \dot a^2}{a^2}=\frac{\ddot a}{a} - \frac{\dot a^2}{a^2} \iff \frac{\ddot a}{a}=\dot H +H^2.
\end{equation}
Consequently, equation \eqref{friedmannwithas} can be rewritten as
\begin{equation}
    - 2 \dot H -2 H^2 - H^2 + 2\dot \omega + 4 H \omega - \omega^2=8 \pi p.
\end{equation}
The final form of the second Friedmann equations \eqref{F2} is therefore obtained
\begin{equation}
    2 \dot{H}+3H^2= - 8 \pi p +4 H \omega - \omega^2 + 2\dot \omega.
\end{equation}

\section{Existence, uniqueness and coordinate-free treatment of semi-symmetric connections}\label{appendixD}

In this Appendix, we present a fully geometric, coordinate-free treatment of semi-symmetric connections.
We introduce the semi-symmetric connections in a formal, coordinate-free way.
\begin{definition}
    On a semi-Riemannian manifold $(M,g)$, a connection $\nabla$ is called \textbf{semi-symmetric} if there exists $\pi \in \Gamma(T^{*}M)$, such that
    \begin{equation}
        \nabla_{X} Y - \nabla_{Y} X - [X,Y]=\pi(Y)X - \pi(X) Y, \forall X,Y \in \Gamma(TM).
    \end{equation}
\end{definition}
\begin{remark}
    Note that in this case, the torsion takes a very specific form, namely:
    \begin{equation}\label{torsionremark}
    \begin{aligned}
        &T(\omega,X,Y)=\omega( \nabla_X Y -\nabla_Y X -[X,Y])\\
        &=\omega \left(\pi(Y)X-\pi(X)Y \right), \; \; \forall X,Y \in \Gamma(TM), \; \; \omega \in \Gamma(T^{*}M).
        \end{aligned}
    \end{equation}
\end{remark}
\begin{proposition}
    In local coordinates, the torsion of a semi-symmetric connection takes the form \eqref{semisymmetric}, i.e.
    \begin{equation}
        \tensor{T}{^\mu _ \nu _\rho}= \pi_{\rho} \delta^{\mu}_{\nu} - \pi_{\nu} \delta^{\mu}_{\rho}.
    \end{equation}
\end{proposition}
\begin{proof}
    We choose $\omega=dx^{\mu}, X=\partial_{\nu},Y=\partial_{\rho}$. Then, according to \eqref{torsionremark}, we have
    \begin{equation}
        \tensor{T}{^\mu _\nu _\rho}=T \left(dx^\mu, \partial_\nu,\partial_\rho \right)=dx^{\mu} \left(\pi(\partial_\rho) \partial_\nu - \pi(\partial_\nu) \partial_\rho \right).
    \end{equation}
   We observe that the components of $\pi$ appear by definition, i.e.
   \begin{equation}
       \tensor{T}{^\mu _\nu _\rho}=dx^{\mu}( \pi_\rho \partial_{\nu} - \pi_{\nu} \partial_{\rho}). \end{equation}
Finally, by using $C^{\infty}(M)$ multilinearity of $dx^{\mu}$ and the notion of dual basis, yields the desired result
\begin{equation}
    \tensor{T}{^\mu _\nu _\rho}=\pi_{\rho} \delta^{\mu}_{\nu} - \pi_{\nu} \delta^{\mu}_{\rho}.
\end{equation}
\end{proof}

We now come to the existence and uniqueness of a semi-symmetric metric connection, a profound result due to Kentaro Yano \cite{Kentaroyano}.
\begin{theorem}[Yano] \label{theoremyano} Let $(M,g,\nabla)$ be a semi-Riemannian manifold with a semi-symmetric connection and denote $\overset{\circ}{\nabla}$ the Levi-Civita connection. If $\nabla$ is metric-compatible, then
\begin{equation}
    \nabla_{X} Y=\overset{\circ}{\nabla}_{X} Y +\pi(Y)X - g(X,Y)P,
\end{equation}
where $P$ is the dual vector field associated to $\pi$, i.e. $g(X,P)=\pi(X)$.
\end{theorem}
\begin{proof}
    Since both $\nabla,\overset{\circ}{\nabla}$ are connections, their difference defines a $(1,2)$-tensor field
    \bea\label{zerotheqn}
        U:\Gamma(TM) &\times& \Gamma(TM) \to \Gamma(TM), \nonumber\\
        U(X,Y)&:=&\nabla_X Y - \overset{\circ}{\nabla}_X Y.
    \eea
    As $\nabla$ is metric-compatible by assumption it follows that
    \bea\label{firsteqn}
        \nabla_{X}(g(Y,Z))&=&g(\nabla_X Y,Z) + g(Y,\nabla_X Z), \nonumber\\
         &&\forall X,Y,Z \in \Gamma(TM).
    \eea

    Similarly, the Levi-Civita connection is metric compatible, which spelled out reads
    \bea\label{secondeqn}
        \overset{\circ}{\nabla}_{X}(g(Y,Z))&=&g\left(\overset{\circ}{\nabla}_X Y,Z \right) + g\left(Y, \overset{\circ}{\nabla}_X Z\right),\nonumber\\
        && \forall X,Y,Z \in \Gamma(TM).
    \eea

    From now on, we will drop the quantifiers, and implicitly understand that the equations we write hold for all vector fields $X,Y,Z$. Since $g(Y,Z)$ is a smooth function, the left hand sides of equations \eqref{firsteqn} \eqref{secondeqn} agree, thus
    \begin{equation}
        0=g(\nabla_X Y,Z) + g(Y,\nabla_X Z) - g\left(\overset{\circ}{\nabla}_X Y,Z \right) -g\left(Y, \overset{\circ}{\nabla}_X Z\right).
    \end{equation}
    Using $C^{\infty}$-multilinearity of the metric, one obtains
\begin{equation}
    0=g \left(\nabla_X Y - \overset{\circ}{\nabla}_X Y,Z \right) + g \left(Y,\nabla_X Z - \overset{\circ}{\nabla}_{X} Z \right).
\end{equation}
By the definition of $U$ from \ref{zerotheqn}, it follows that
\begin{equation}
    g(U(X,Y),Z)+g(Y,U(X,Z))=0,
\end{equation}
or
\begin{equation}\label{keepinmind}
    g(U(X,Y),Z)+g(U(X,Z),Y)=0.
\end{equation}
Similarly, from the definition of $U$, one has
\bea
    U(X,Y)- U(Y,X)&=&\nabla_X Y - \overset{\circ}{\nabla}_X Y - \nabla_Y X + \overset{\circ}{\nabla}_{Y} X\nonumber\\
    &=&\underbrace{\nabla_X Y - \nabla_Y X -[X,Y]}_{=T(X,Y)},
\eea

Since the Levi-Civita connection is torsion free, from which we find
\begin{equation}
    \begin{aligned}
        g(T(X,Y),Z)&=g(U(X,Y),Z)-g(U(Y,X),Z),\\
        g(T(Z,X),Y)&=g(U(Z,X),Y)-g(U(X,Z),Y),\\
        g(T(Z,Y),X)&=g(U(Z,Y),X)-g(U(Y,Z),X).
    \end{aligned}
\end{equation}
Adding the three equations, and  keeping in mind \eqref{keepinmind}, we find
\begin{equation}
\begin{aligned}
    g(T(X,Y),Z)&+g(T(Z,X),Y)\\
    &+g(T(Z,Y),X)=2g(U(X,Y),Z).
\end{aligned}
\end{equation}
By non-degeneracy of $g$, one can read off $U$, i.e.
\begin{equation}\label{asd5}
    U(X,Y)=\frac{1}{2} \left(T(X,Y)+T'(X,Y)+T'(Y,X) \right),
\end{equation}
where $T'$ is a type $(1,2)$ tensor defined as
\begin{equation}\label{tprime}
    g(T(Z,X),Y)=g(T'(X,Y),Z).
\end{equation}
By the definition of $U$ c.f. \eqref{zerotheqn}, we obtain:
\begin{equation}
    \frac{1}{2} \left(T(X,Y)+T'(X,Y)+T'(Y,X) \right)=\nabla_X Y - \overset{\circ}{\nabla}_{X} Y,
\end{equation}
from which expressing $\nabla_X Y$ leads to:
\begin{equation} \label{asd}
    \nabla_X Y=\overset{\circ}{\nabla}_{X} Y +\frac{1}{2} \left(T(X,Y)+T'(X,Y)+T'(Y,X) \right).
\end{equation}
By assumption, $\nabla$ is semi-symmetric, which upon substitution in \eqref{tprime} yields
\begin{equation}
    g(\pi(X) Z - \pi(Z) X,Y)=g(T'(X,Y),Z),
\end{equation}
or by linearity
\begin{equation}\label{asd4}
    g(\pi(X)Z,Y)-g(\pi(Z)X,Y)=g(T'(X,Y),Z).
\end{equation}
By using $C^{\infty}$ bilinearity and symmetry of the metric, it follows that
\begin{equation}\label{asd2}
\begin{aligned}
    g(\pi(X)Z,Y)&=\pi(X) g(Z,Y)\\
    &=\pi(X) g(Y,Z)=g(\pi(X)Y,Z).
\end{aligned}
\end{equation}

We similarly have for the second term, using the definition of the dual vector field $P$
\begin{equation}\label{asd3}
   \begin{aligned} g(\pi(Z)X,Y)&=g(g(Z,P)X,Y)=g(Z,P)g(X,Y)\\
   &=g(P,Z)g(X,Y)=g(g(X,Y)P,Z).
\end{aligned}
\end{equation}

By substituting \eqref{asd2} and \eqref{asd3} into \eqref{asd4}, we obtain
\begin{equation}
    g(\pi(X)Y,Z)-g(g(X,Y)P,Z)=g(T'(X,Y),Z),
\end{equation}
from which by non-degeneracy of the metric
\begin{equation}
    T'(X,Y)=\pi(X)Y-g(X,Y)P.
\end{equation}
Plugging in the torsion $T$ and the obtained $T'$ into \eqref{asd5} tells us that $U$ is given by
\begin{equation}
\begin{aligned}
    U(X,Y)&=\frac{1}{2}\left(\pi(Y)X-\pi(X)Y +\pi(X)Y \right)\\
    &+\frac{1}{2} \left(-g(X,Y)P+\pi(Y)X-g(Y,X)P\right),
\end{aligned}
\end{equation}
or upon simplification
\begin{equation}
    U(X,Y)=\pi(Y)X - g(X,Y)P.
\end{equation}
Substituting in the above result into \eqref{zerotheqn} leads to
\begin{equation}
    \pi(Y)X - g(X,Y)P=\nabla_{X}Y -\overset{\circ}{\nabla}_{X} Y,
\end{equation}
which is equivalent to the desired result
\begin{equation}\label{yanoresult}
 \nabla_{X}Y=\overset{\circ}{\nabla}_{X} Y + \pi(Y) X - g(X,Y)P.
\end{equation}
\end{proof}

We present a direct corollary, which we heavily use in the main text of the article.
\begin{corollary}
    In a local coordinate system, the Christoffel symbols of a semi-symmetric metric connection take the form \eqref{Christoffelsemisymmetric}
   \begin{equation}
     \tensor{{\Gamma}}{^\mu _\nu _\rho}= \tensor{\gamma}{^\mu _\nu _\rho} - \pi^{\mu} g_{\rho \nu} + \pi_{\nu} \delta^{\mu}_{\rho}.
\end{equation}
\end{corollary}
\begin{proof}
    We choose $X=\partial_{\rho}, Y=\partial_{\nu}$. First, we compute the coordinates of the dual vector field. We start from the definition
    \begin{equation}
        g(\partial_{\rho},P^\mu \partial_{\mu})=\pi(\partial_{\rho}).
    \end{equation}
    It immediately follows that
    \begin{equation}
        P^{\mu} g_{\rho \mu}=\pi_{\rho} \iff P^{\mu} g^{\beta \rho} g_{\rho \mu}=\pi_{\rho} g^{\beta \rho},
    \end{equation}
    which means that the components of the dual vector field are the same as the components of the raised index $\pi$, as expected:
    \begin{equation}
        P^{\beta}=\pi^{\beta}.
    \end{equation}
    By substituting the local expressions in \eqref{yanoresult}, one obtains
    \begin{equation}
        \nabla_{\partial_\rho} \partial_{\nu}=\overset{\circ}{\nabla}_{\partial_\rho} \partial_{\nu} + \pi(\partial_\nu) \partial_\rho - g(\partial_\rho,\partial_\nu) \pi^{\sigma} \partial_{\sigma},
    \end{equation}

    From the definition of Christoffel symbols and components of a tensor field
    \begin{equation}
        \tensor{\Gamma}{^\beta_\nu _\rho} \partial_{\beta}= \tensor{\gamma}{^\beta _\nu _\rho} \partial_{\beta} + \pi_{\nu} \partial_{\rho} - g_{\rho \nu} \pi^{\sigma} \partial_{\sigma}
    \end{equation}

    To extract the Christoffel symbols, we apply a one-form $dx^{\mu}$ to the equation and use that it is a dual basis to the coordinate vector fields. This way, the desired result
    \begin{equation}
        \tensor{\Gamma}{^\mu _\nu _\rho}=\tensor{\gamma}{^\mu _\nu _\rho} + \pi_{\nu} \delta^{\mu}_{\rho} - g_{\rho \nu} \pi^{\mu}
    \end{equation}
    is obtained.
\end{proof}

Using very similar techniques to the proof presented in Theorem~\ref{theoremyano}, Kentaro Yano \cite{Kentaroyano} also proved a coordinate-free Theorem concerning the curvature tensor of a semi-symmetric metric connection. First of all, let us fix the convention for the Riemann curvature tensor.
\begin{definition}
    Let $(M,g)$ be a semi-Riemannian manifold with an affine connection $\nabla$. The Riemann curvature tensor $Riem$ is a $(1,3)$-tensor field defined as
    \begin{equation*}
        Riem(\omega,Z,X,Y)=\omega \left(\nabla_{X} \nabla_{Y}Z-\nabla_{Y} \nabla_{X} Z - \nabla_{[X,Y]}Z \right)
    \end{equation*}
\end{definition}
In this convention, the components in coordinates are given by the following proposition, whose proof is standard.
\begin{proposition}
    The components of $Riem$ in a local chart given by $\omega=dx^{\mu},Z=\partial_{\nu},X=\partial_{\rho},Y=\partial_{\sigma}$ are given by
   \begin{equation}
   \begin{aligned}
   \tensor{Riem}{^\mu _\nu _\rho _\sigma}&=\tensor{\Gamma}{^\lambda _\nu _\sigma} \tensor{\Gamma}{^\mu _\lambda _\rho} - \tensor{\Gamma}{^\lambda _\nu _\rho} \tensor{\Gamma}{^\mu_\lambda _\sigma}\\
   &+ \partial_{\rho} \tensor{\Gamma}{^\mu _\nu _\sigma} - \partial_{\sigma} \tensor{\Gamma}{^\mu _\nu _\rho}.
   \end{aligned}
   \end{equation}
\end{proposition}
We now quote Kentaro Yano's coordinate-free theorem, which gives the relation between the curvature tensors of the semi-symmetric connection and the Levi-civita connection. We do not write out the proof, as it is very similar to that one of Theorem ~\ref{theoremyano}.
\begin{theorem}[Yano]
    Let $(M,g,\nabla)$ be a semi-Riemannian manifold with a semi-symmetric connection that is metric-compatible and denote $\overset{\circ}{\nabla}$ the Levi-Civita connection. Moreover, denote by $Riem$ the curvature tensor of the semi-symmetric metric connection $\nabla$ and by $\overset{\circ}{Riem}$ the curvature tensor of the Levi-Civita connection $\overset{\circ}{\nabla}$. Then, the following equation
    \begin{equation}\label{coordinatefreecurvature}
    \begin{aligned}
        Riem(\omega,Z,X,Y)&=\overset{\circ}{Riem}(\omega,Z,X,Y)-\omega( S(Y,Z) X)\\
        &+\omega(S(X,Z)Y)
        - \omega(g(Y,Z) A(X))\\
        &+\omega(g(X,Z)A(Y)),
    \end{aligned}
    \end{equation}
   is satisfied for all one-forms $\omega$, and vector fields $X,Y,Z$, where
    \begin{equation}
        S(X,Y)=\left(\overset{\circ}{\nabla}_{X} \pi \right)(Y)- \pi(X) \pi(Y) + \frac{1}{2} \pi(P) g(X,Y)
    \end{equation}
    and $A$ is a $(1,1)-$tensor field defined by
    \begin{equation}
        g(A(X),Y)=S(X,Y).
    \end{equation}
\end{theorem}

This result is relevant, because with the help of it, we can justify the form of the Riemann curvature tensor \eqref{riemanncurvaturesemisym} using the upcoming corollary.
\begin{corollary}\label{corollaryriemann}
    In a local coordinate system, the Riemann tensor of a semi-symmetric metric connection takes the form
    \begin{equation}
  \begin{aligned}
    \tensor{Riem}{^\mu_{\nu \rho \sigma}}=&\overset{\circ}{Riem} \tensor{}{^\mu _\nu _\rho _\sigma}- S_{\sigma \nu} \delta^{\mu}_{\rho}+S_{\rho \nu} \delta^{\mu}_{\sigma}\\
    &- g_{\sigma \nu} S_{\rho \lambda} g^{\lambda \mu}+g_{\rho \nu} S_{\sigma \lambda} g^{\lambda \mu}.
\end{aligned}
\end{equation}
\end{corollary}

\begin{proof}

    We choose a local coordinate system given by $\omega=dx^{\mu},Z=\partial_{\nu},X=\partial_{\rho},Y=\partial_{\sigma}$. Upon this choice, we compute the coordinates of the dual vector field $A$:
    \begin{equation}
    g\left(\left(A(\partial_{\nu})\right)^{\beta} \partial_{\beta},\partial_{\rho} \right)=S(\partial_{\nu}, \partial_{\rho}),
    \end{equation}
    \begin{equation}\label{partialresult}
        A_{\nu}^{\beta} g_{\beta \rho}=S_{\nu \rho} \iff A^{\mu}_{\nu}=S_{\nu \lambda} g^{\mu \lambda}.
    \end{equation}
 We now evaluate \eqref{coordinatefreecurvature}
    \bea
    &&R \left( dx^{\mu},\partial_{\nu},\partial_{\rho}, \partial_{\sigma} \right)= \overset{\circ}{R}\left( dx^{\mu},\partial_{\nu},\partial_{\rho}, \partial_{\sigma} \right) - dx^{\mu}( S(\partial_{\sigma},\partial_{\nu}) \partial_{\rho})\nonumber\\
       &&+ dx^{\mu}(S(\partial_{\rho},\partial_\nu) \partial_\sigma)- dx^{\mu}\left(g(\partial_{\sigma}, \partial_\nu) \left(A(\partial_{\rho}) \right)^{\beta} \partial_{\beta}\right)\nonumber\\
       &&+ dx^{\mu}\left( g(\partial_{\rho}, \partial_\nu) \left(A(\partial_\sigma) \right)^{\beta} \partial_{\beta} \right).
        \eea

    Using the definition of components, $C^{\infty}(M)$ multilinearity, and the notion of a dual basis, one obtains
    \bea
    \tensor{R}{^\mu _\nu _\rho _\sigma}=\overset{\circ}{R}\tensor{}{^\mu _\nu _\rho _\sigma} &-& S_{\sigma \nu} \delta^{\mu}_\rho + S_{\rho \nu} \delta^{\mu}_{\sigma}  \nonumber\\
        &-& g_{\sigma \nu}A_{\rho}^{\beta} \delta^{\mu}_{\beta} + g_{\rho \nu} A^{\beta}_{\sigma} \delta^{\mu}_{\beta}.
    \eea

    Substituting the partial result \eqref{partialresult} leads to the desired formula
    \bea
    \tensor{R}{^\mu _\nu _\rho _\sigma}=\overset{\circ}{R}\tensor{}{^\mu _\nu _\rho _\sigma} &-& S_{\sigma \nu} \delta^{\mu}_\rho + S_{\rho \nu} \delta^{\mu}_{\sigma}  \nonumber\\
        & -& g_{\sigma \nu} S_{\rho \lambda} g^{\mu \lambda} +g_{\rho \nu} S_{\sigma \lambda} g^{\mu \lambda}.
       \eea
\end{proof}

\end{appendix}


\begin{thebibliography}{99}


\bibitem{Einstein} A. Einstein, Die Feldgleichungen der Gravitation, Sitzungsberichte der Königlich Preussischen Akademie der Wissenschaften zur Berlin, 844 (1915).

\bibitem{Hilbert} D. Hilbert, Die Grundlagen der Physik, Nachrichten von der Gesellschaft der Wissenschaften zu Göttingen - Mathematisch- Physikalische Klasse \textbf{3}, 395 (1915).

\bibitem{Einstein2} A. Einstein, Die Grundlage der allgemeinen Relativit\"{a}tstheorie, Annalen der Physik \textbf{354}, 769 (1916).

\bibitem{Will} C. M. Will, The Confrontation between General Relativity and Experiment, Living Reviews in Relativity {\bf 17}, 4 (2014).

\bibitem{Gravwaves} B. P. Abbott et al. (LIGO Scientific Collaboration and Virgo Collaboration), Observation of Gravitational Waves from a Binary Black Hole Merger, Physical Review Letters {\bf 116} 061102 (2016).

\bibitem{Weyl} H. Weyl, Gravitation und Elektrizit\"{a}t, Sitzungsberichte der K\"{o}niglich Preussischen Akademie der Wissenschaften zu Berlin, 465 (1918).

\bibitem{Weyl1} H. Weyl, Space-time-matter, Dover, New York, 1952

\bibitem{Collectedpapers}  Schulman, R., Kox, A.J., Janssen, M., and Illy, J., eds., Einstein Collected Papers, Vol. 8A-8B: The Berlin Years: Correspondence, 1914–1918, (Princeton University Press, Princeton,
NJ, 1997), Volume \textbf{8B}, Document \textbf{544}, 767.

\bibitem{Cartan1} \'{E}. Cartan,  Sur une g\'{e}n\'{e}ralisation de la notion de courbure de Riemann et les espaces
\`{a} torsion, C. R. Acad. Sci. (Paris) \textbf{174}, 593 (1922).

\bibitem{Cartan2} \'{E}. Cartan,  Sur les vari\'{e}t\'{e}s \`{a} connexion affine et la theorie de la relativit\'{e} g\'{e}n\'{e}ralis\'{e}e, Ann. \'{E}c. Norm. Sup. \textbf{40}, 325 (1923).

\bibitem{Cartan3} \'{E}. Cartan, Sur les vari\'{e}t\'{e}s \`{a} connexion affine et la theorie de la relativit\'{e} g\'{e}n\'{e}ralis\'{e}e, Ann. \'{E}c. Norm. Sup. \textbf{41}, 1 (1924)

\bibitem{Cartan4} \'{E}. Cartan, Ann. \'{E}c. Norm. Sup., Sur les vari\'{e}t\'{e}s \`{a} connexion affine et la theorie de la relativit\'{e} g\'{e}n\'{e}ralis\'{e}e,  \textbf{42}, 17 (1925)

\bibitem{Hehl} F. W. Hehl,  P. von der Heyde, D. G. Kerlick, and J. M. Nester, General relativity with spin and torsion: Foundations and prospects, Reviews of Modern Physics  {\bf 48}, 393 (1976).

    \bibitem{Hehl1} V. De Falco, and E. Battista, Analytical results for binary dynamics at the first post-Newtonian order in Einstein-Cartan theory with the Weyssenhoff fluid, Phys. Rev. D {\bf 108}, 064032 (2023).

\bibitem{Hehl2} M. Piani and J. Rubio, Preheating in Einstein-Cartan Higgs Inflation: Oscillon formation, Journal of Cosmology and Astroparticle Physics {\bf 2023}, 002 (2023).

\bibitem{Hehl3} W. Barker and S. Zell, Einstein-Proca theory from the Einstein-Cartan formulation, Phys. Rev. D {\bf 109}, 024007 (2024).

\bibitem{WC1} H.-H. von Borzeszkowski and H.-J. Treder, The Weyl-Cartan Space Problem in Purely Affine Theory, Gen. Rel. Grav. {\bf 29},  455 (1997).

\bibitem{WC2} D. Puetzfeld and R. Tresguerres, A cosmological model in Weyl-Cartan spacetime, Class. Quant.
Grav. {\bf 18},  677 (2001).

\bibitem{WC3} D. Putzfeld, A cosmological model in Weyl-Cartan spacetime: I. Field equations and solutions, Class. Quant. Grav. {\bf 19},  3263 (2002).

\bibitem{WC4} M. Novello and S. E. Perez Bergliaffa, Bouncing cosmologies, Physics Reports {\bf 463} 127  (2008).

\bibitem{WC5} J. Attard, J. François, and S. Lazzarini, Weyl gravity and Cartan geometry, Phys. Rev. D {\bf 93}, 085032 (2016).

\bibitem{WC6} T. Harko, N. Myrzakulov, R. Myrzakulov, and S. Shahidi, Non-minimal geometry-matter couplings in Weyl-Cartan space-times: $f\left(R,T,Q,T_m\right)$ gravity, Phys. Dark Universe {\bf 4}, 100886 (2021).

\bibitem{WC7} S. Bahamonde and J. G. Valcarcel, Algebraic classification of the gravitational field in Weyl-Cartan space-times, Phys. Rev. D {\bf 108} 044037 (2023).

\bibitem{WC8} J. Wang, L.-X. Qiang, Y.-F. Zhao, Q.-Y. Yin, and X.-Y. Chen, Constraints on cosmological model in Weyl-Cartan spacetime from astronomical measurements, Physica Scripta {\bf 98}, 115034 (2023).

\bibitem{Weitz} R. Weitzenb\"{o}ck, Invariantentheorie, Noordhoff, Groningen (1923).
\bibitem{Ein} A. Einstein, Riemanngeometrie mit Aufrechterhaltung des Begriffes des Fern-Parallelismus, Preussische Akademie der Wissenschaften, Phys.-math. Klasse, Sitzungsberichte {\bf 1928} 217 (1928).

\bibitem{TP1} K. Hayashi and T. Shirafuji, New general relativity, Phys. Rev. D {\bf  19},  3524 (1979).
\bibitem{TP2} R. Ferraro and F. Fiorini, Modified teleparallel gravity: Inflation without an inflaton, Phys. Rev. D {\bf 75},  084031 (2007).

\bibitem{TP3}  C. G. Boehmer, T. Harko, and F. S. N. Lobo, Wormhole geometries in modified teleparallel gravity and the energy conditions, Phys. Rev. D {\bf 85},  044033 (2012).

\bibitem{TP4} Z. Haghani, T. Harko, H. R. Sepangi, S. Shahidi, Weyl-Cartan-Weitzenb\"{o}ck gravity as a generalization of teleparallel gravity, JCAP {\bf 10},  061 (2012).

\bibitem{TP5} Z. Haghani, T. Harko, H. R. Sepangi, and S. Shahidi, Weyl-Cartan-Weitzenb\"{o}ck gravity through Lagrange multiplier,
Phys. Rev. D {\bf 88},  044024 (2013).

\bibitem{TP6} T.  Harko, F. S. N. Lobo, G. Otalora, and E. N. Saridakis, $f(T,\cal{T})$ gravity and cosmology, Journal of Cosmology and Astroparticle Physics {\bf  2014}, 021 (2014).

 \bibitem{TP7} J. Levi Said, J. Mifsud, J. Sultana, and K. Zarb Adami, Reconstructing teleparallel gravity with cosmic structure growth and expansion rate data, Journal of Cosmology and Astroparticle Physics {\bf 2021}, 015 (2021).

\bibitem{TP8} M. Chakrabortty, N. Sk, and A. K. Sanyal, A viable form of the metric Teleparallel F(T) theory of gravity, Eur. Phys. J. C {\bf 83}, Issue 7, 557 (2023).

\bibitem{TP9} Y.-M. Hu, Y. Zhao, X. Ren, B. Wang, E. N. Saridakis, and Y.-F. Cai, The effective field theory approach to the strong coupling issue in $f(T)$ gravity, Journal of Cosmology and Astroparticle Physics {\bf 2023}, 060 (2023).

\bibitem{TP10} S. S. Mishra, S. Mandal, and P. K.  Sahoo, Constraining $f (T , \cal{T})$ gravity with gravitational baryogenesis, Phys. Lett. B {\bf 842},  137959 (2023).

\bibitem{TP11} J. C. N. de Araujo and H. G. M. Fortes, Compact stars in $f(T)=T+\xi T^\beta$ gravity, Eur. Phys. J. C {\bf 83}, 1168 (2023).

\bibitem{TPRev} S. Bahamonde et al., Teleparallel gravity: from theory to cosmology, Reports on Progress in Physics {\bf 86}, 026901 (2023).

\bibitem{Schrod} E. Schr\"{o}dinger, Space-Time Structure (Cambridge Science Classics). Cambridge: Cambridge University Press,
1985

\bibitem{Edd} A. S. Eddington, The Mathematical Theory of Relativity, Cambridge Univdersity Press, Cambridge, 1923

\bibitem{unified} H. F.M. Goenner, On the History of Unified Field Theories, Living Reviews in Relativity {\bf  7},  (2004)

\bibitem{Planck}  N. Aghanim et al. Planck Collaboration, Planck 2018 results
VI. Cosmological parameters, Astronomy and
Astrophysics {\bf 641}, A6 (2020).

\bibitem{Baryon1} K. S. Dawson et al., The Baryon Oscillation Spectroscopic Survey of SDSS-III, Astron. J. \textbf{145}, 10 (2013).

\bibitem{Baryon2} K. S. Dawson et al., The SDSS-IV Extended Baryon Oscillation Spectroscopic Survey: Overview and Early Data, Astron. J. \textbf{151}, 44 (2016).

\bibitem{Baryon3}  M. Gatti et al., Dark Energy Survey Year 3 Results: clustering redshifts - calibration of the weak lensing source redshift distributions with redMaGiC and BOSS/eBOSS, Monthly Notices of the Royal Astronomical Society \textbf{510}, 1223 (2022).

 \bibitem{acceleration} D. H. Weinberg, M. J. Mortonson, D. J. Eisenstein, C.
Hirata, A. G. Riess, and E. Rozo, Observational probes of cosmic acceleration, Physics Reports \textbf{530},
87 (2013).

\bibitem{rotationcurves1} P. Salucci, C. Frigerio Martins, and A. Lapi, DMAW 2010 LEGACY the Presentation Review: Dark Matter in Galaxies with its Explanatory Notes,
arXiv:1102.1184 (2011).

\bibitem{rotationcurves2} M. Persic, P. Salucci, and F. Stel, The universal rotation curve of spiral galaxies — I. The dark matter connection, Mon. Not. R. Astron.
Soc. \textbf{281}, 27 (1996).

\bibitem{rotationcurves3} A. Boriello and P. Salucci, The dark matter distribution in disc galaxies, Mon. Not. R. Astron. Soc.
\textbf{323}, 285 (2001).

\bibitem{darkmattercandidate1}  J. M. Overduin and P. S. Wesson, Dark matter and background light, Phys. Repts. \textbf{402},
267 (2004).

\bibitem{darkmattercandidate2} L. Bian, X. Liu, and K.-P. Xie, Probing superheavy dark matter with gravitational waves, Journal of High Energy
Physics \textbf{2021}, 175 (2021).

\bibitem{fQ1}  J.M. Nester and  H-J Yo, Symmetric teleparallel general relativity, Chinese Journal of Physics {\bf 37}, 113 (1999).

\bibitem{fQ2}  J. B. Jimenez, L. Heisenberg, and T. Koivisto,  Coincident general relativity, Phys. Rev.
D {\bf 98}, 044048, (2018).

\bibitem{fQ3} T. Harko, T. S. Koivisto, F. S. N. Lobo, G. J. Olmo, and D. Rubiera-Garcia, Coupling matter in modified f(Q) gravity, Phys. Rev. D {\bf 98}, 084043, (2018).

\bibitem{fQ4} Y. Xu, G. Li, T. Harko, and S.-D. Liang, $f(Q, T)$ gravity, Eur. Phys. J. C {\bf 79}, 708 (2019).

\bibitem{fQ5}  Y. Xu, T. Harko, S. Shahidi, and S.-D. Liang, Weyl type $f(Q, T)$ gravity, and its cosmological implications, Eur. Phys. J. C {\bf 80}, 449 (2020).

\bibitem{fQ6} N. Frusciante, Signatures of $f(Q)$ gravity in cosmology, Phys. Rev. D {\bf 103}, 044021, (2021).

\bibitem{fQ7}  R. H. Lin and X. H. Zhai, Spherically symmetric configuration in $f(Q)$ gravity, Phys. Rev. D {\bf 103}, 124001, (2021).

\bibitem{fQ8} S. V. Lohakare, S. K. Maurya, Ksh. Newton Singh, B. Mishra, and A. Errehymy, Influence of three parameters on maximum mass and stability of strange star under linear $f(Q)$-action, Monthly Notices of the Royal Astronomical Society {\bf 526}, 3796 (2023).

\bibitem{fQ9} A. Mussatayeva, N. Myrzakulov, and M. Koussour, Cosmological constraints on dark energy in f(Q) gravity: A parametrized perspective, Physics of the Dark Universe {\bf 42}, article  101276 (2023).

\bibitem{fQ10} S. Mandal, S. Pradhan, P. K. Sahoo, and T. Harko, Cosmological observational constraints on the power law $f(Q)$ type modified gravity theory, Eur. Phys. J. C {\bf 83}, 1141 (2023).

\bibitem{fQ11} P. Bhar, Ksh. Newton Singh, S. K. Maurya, and M. Govender, A four parameters quark star in quadratic $f(Q)$ - action, Physics of the Dark Universe {\bf 43}, 101391 (2024).

\bibitem{Gh1}  D. M. Ghilencea, Gauging scale symmetry and inflation: Weyl versus Palatini gravity, Eur. Phys. J. C {\bf 81}, 510 (2021).

\bibitem{Gh2}  D. M. Ghilencea, Standard Model in Weyl conformal geometry, Eur. Phys. J. C {\bf 82}, 23 (2022).

\bibitem{Gh3}  D. M. Ghilencea, Non-metric geometry as the origin of mass in gauge theories of scale invariance, Eur. Phys. J. C {\bf 83}, 176 (2023).

\bibitem{Gh4} M. Weisswange, D. M. Ghilencea, and D. St\"{o}ckinger, Quantum scale invariance in gauge theories and applications to muon production,
Phys. Rev. D {\bf 107}, 085008 (2023).

\bibitem{Gh5} D. M. Ghilencea and C. T. Hill, Renormalization group for nonminimal $\phi ^2R$ couplings and gravitational contact interactions, Phys. Rev. D {\bf 107}, 085013 (2023).

\bibitem{Harkoweylastro} J.-Z. Yang, S. Shahidi, and T. Harko, Black hole solutions in the quadratic Weyl conformal geometric theory of gravity, Eur. Phys. J. C {\bf 82}, 1171 (2022).

\bibitem{Harkodarkmatter}
P. Burikham, T. Harko, K. Pimsamarn, and S. Shahidi, Dark matter as a Weyl geometric effect, Phys. Rev. D. \textbf{107}, 064008 (2023).

\bibitem{Harkodarkmattera} M. Craciun and T. Harko, Testing Weyl geometric gravity with the SPARC galactic rotation curves database, to be published in Physics of the Dark Universe (2024).

    \bibitem{Gh6} T. Harko and S. Shahidi, Coupling matter and curvature in Weyl geometry: conformally invariant $f\left(R,L_m\right)$ gravity, Eur. Phys. J. C {\bf 82}, 219 (2022).

\bibitem{Gh7} T. Harko and S. Shahidi, Palatini formulation of the conformally invariant $f\left(R,L_m\right)$ gravity theory, Eur. Phys. J. C {\bf 82}, 1003 (2022).

\bibitem{HarkoSchr}
 L. Ming, S.-D. Liang, H.-H. Zhang, and T. Harko, From the Weyl-Schr\"{o}dinger connection to the accelerating Universe: Extending Einstein's gravity via a length preserving nonmetricity,
 Phys. Rev. D. \textbf{109}, 024003 (2024).

 \bibitem{SilkeKlemm}
 S. Klemm and L. Ravera, Schr\"{o}dinger connection with selfdual nonmetricity vector in 2+1 dimensions, Phys. Lett. B {\bf 817}, 136291 (2021).

\bibitem{FriedmannSchouten} A. Friedmann and J. A. Schouten, \"{U}ber die Geometrie der halbsymmetrischen \"{u}bertragung, Math. Zeitschr. \textbf{21}, 211 (1924).

\bibitem{Kentaroyano}
K. Yano,   On Semi-Symmetric Metric Connection, Revue Roumaine de Math\'{e}matique Pures et Appliqu\'{e}es {\bf 15}, 1579 (1970).

\bibitem{zangiabadi}
E. Zangiabadi and Z. Nazari, Semi-Riemannian manifold with semi-symmetric connections, Journal of Geometry and Physics {\bf 169}, 104341 (2021).

\bibitem{Riccicalculusbook}
J.A. Schouten, Ricci Calculus, An Introduction to Tensor Analysis and Geometrical Applications, Springer-Verlag, Berlin-Göttingen-Heidelberg, 1954

\bibitem{unificationsemisym} Gh. Fasihi-Ramandi, Semi-Symmetric connection formalism for unification of gravity and electromagnetism,
Journal of Geometry and Physics {\bf  144}, 245 (2019)

\bibitem{Kranas} D. Kranas, C. G. Tsagas, J. D. Barrow, and D. Iosifidis, Friedmann-like universes with torsion, Eur. Phys. J. C {\bf 79},  341 (2019).

\bibitem{Ben1} D. Benisty, A. van de Venn, D. Vasak, J. Struckmeier, and H. Stoecker, Torsional dark energy, Int. J. Mod. Phys. D {\bf 31}, 2242013-109 (2022).

\bibitem{Ben2} D. Benisty, E. I. Guendelman, A. van de Venn, D. Vasak, J. Struckmeier, and H. Stoecker, The dark side of the torsion: Dark Energy from propagating torsion, Eur. Phys. J. C {\bf 82}, 264 (2022).

\bibitem{Luz} P. Luz and  J. P. S. Lemos, Relativistic cosmology and intrinsic spin of matter: Results and theorems in Einstein-Cartan theory, Phys. Rev. D {\bf 107},  084004 (2023). 

\bibitem{murgescu}
V. Murgescu, Espaces de Weyl a Torsion et Leurs Representations Conformes, Ann. Sci. Univ. Timisoara, 19,221-228 (1968).

\bibitem{unaluysal}
F. \"{U}nal and A. Uysal, Weyl Manifolds with Semi-Symmetric Connection, Math. Comput. Appl.  {\bf 10}, 351 (1970).

\bibitem{1g} Y. Akrami et al., Planck 2018 results. I. Overview and the
cosmological legacy of Planck, Astron. Astrophys. \textbf{641}, A1 (2020).


\bibitem{Bou} A. Bouali, H. Chaudhary, T. Harko, F. S. N. Lobo, T. Ouali, and M. A. S. Pinto, Observational Constraints and Cosmological Implications of Scalar-Tensor $f(R,T)$ Gravity, MNRAS {\bf 526}, 4192 (2023).

\bibitem{Sahni} V. Sahni, A. Shafieloo and A. A. Starobinsky, Two new diagnostics of dark energy, Phys. Rev. D \textbf{78}, 103502 (2008).

\bibitem{Pol} M. Chevallier and D. Polarski, Accelerating Universes with Scaling Dark Matter, Int. J. Mod. Phys. D {\bf 10},
213 (2001).
\bibitem{Pol1}  E. V. Linder, Exploring the Expansion History of the Universe, Phys. Rev. Lett. {\bf 90}, 091301 (2003).

\bibitem{G3} J. Beltran Jimenez, L. Heisenberg, and T. S. Koivisto, The Geometrical Trinity of Gravity, Universe {\bf 5}, 173 (2019).

\bibitem{fRLm} T. Harko and F. S. N. Lobo, $f\left(R,L_m\right)$ gravity, Eur. Phys. J. C {\bf 70}, 373 (2010).

\bibitem{fRT} T. Harko, F. S. N. Lobo, S. Nojiri and S. D. Odintsov, $f(R,T)$ gravity, Phys.
Rev. \textbf{D 84}, 024020 (2011).

\bibitem{book} T. Harko and F. S. N. Lobo, Extensions of $f(R)$ Gravity:
Curvature-Matter Couplings and Hybrid Metric Palatini Theory, Cambridge
University Press, Cambridge, UK, 2018

\bibitem{Parker} L. Parker, Particle Creation in Expanding Universes, Phys. Rev. Lett. \textbf{21}, 562 (1968).

\bibitem{Parker1} L. Parker, Quantized Fields and Particle Creation in Expanding Universes. I, Phys. Rev. \textbf{183}, 1057 (1969).

\bibitem{Parker2} S. A. Fulling, L. Parker and B. L. Hu, Conformal energy-momentum tensor in curved spacetime: Adiabatic regularization and renormalization, Phys. Rev. \textbf{10}, 3905 (1974).

\bibitem{P-M} I. Prigogine, J. Geheniau, E. Gunzig, and P. Nardone, Thermodynamics of Cosmological Matter Creation,
Proceedings Of The National Academy Of Sciences \textbf{85}, 7428 (1988).

\bibitem{Lima} M. O. Calvao, J. A. S. Lima, and I. Waga, On the thermodynamics of matter creation in cosmology, Phys. Lett. A
\textbf{162}, 223 (1992).

\bibitem{Su} J. Su, T. Harko, and S.-D. Liang, Irreversible thermodynamic description of dark matter and radiation creation during inflationary reheating, Advances in High Energy Physics \textbf{2017}, 7650238 (2017).

\bibitem{Bar} J. A. S. Lima and I. P. Baranov, Gravitationally induced particle production: Thermodynamics and kinetic theory, Phys. Rev. D \textbf{90},
043515 (2014).

\end{thebibliography}
\end{document}